\documentclass[preprint]{elsarticle}

\usepackage{enumerate}
\usepackage{graphicx}
\usepackage{epstopdf}
\usepackage{amsfonts}
\usepackage{amssymb}
\usepackage{amsmath}
\usepackage{wasysym}
\usepackage{hyperref}
\usepackage{mathtools}
\usepackage{comment}

\newtheorem{theorem}{Theorem}
\newtheorem{lemma}{Lemma}
\newtheorem{proposition}{Proposition}

\newtheorem{example}{Example}

\newdefinition{definition}{Definition}
\newproof{proof}{Proof}

\newcommand{\R}{\mathbb{R}}

\newcommand{\K}{\mathcal{K}}
\newcommand{\X}{\mathcal{X}}
\newcommand{\Co}{C^n}
\newcommand{\Cv}{C^n(v)}
\renewcommand{\P}{P^n}
\newcommand{\Pv}{P^n(v)}
\newcommand{\cif}{\text{if }}

\newcommand{\cand}{\text{ and }}

\begin{document}

\begin{frontmatter}

\title{Positivity and convexity in incomplete cooperative games\tnoteref{funding}
}

\tnotetext[funding]{The first and the second author were supported by SVV--2020--260578 and Charles University Grant Agency (GAUK 341721).
The first author acknowledges the support of Student Faculty Grant by Faculty of Mathematics and Physics of Charles University.}

\author[1]{Martin Černý}
\ead{cerny@kam.mff.cuni.cz}

\author[2]{Jan Bok\corref{cor1}}
\cortext[cor1]{Corresponding author}
\ead{bok@iuuk.mff.cuni.cz}

\author[2,3]{David Hartman}
\ead{hartman@iuuk.mff.muni.cz}

\author[1]{Milan Hladík}
\ead{hladik@kam.mff.cuni.cz}

\address[1]{Department of Applied Mathematics, Faculty of Mathematics and Physics, Charles University, Prague, Czech Republic}
\address[2]{Computer Science Institute, Faculty of Mathematics and Physics, Charles University, Prague, Czech Republic}
\address[3]{Institute of Computer Science of the Czech Academy of Sciences, Prague, Czech Republic}


\begin{abstract}
Incomplete cooperative games generalise the classical model of cooperative games by omitting the values of some of the coalitions. This allows to incorporate uncertainty into the model and study the underlying games as well as possible payoff distribution based only on the partial information. In this paper we perform a systematic study of incomplete games, focusing on two important classes of cooperative games: positive and convex games.

Regarding positivity, we generalise previous results for a special class of minimal incomplete games to general setting.  We characterise non-extendability to a positive game by the existence of a certificate and provide a description of the set of positive extensions using its extreme games. The results are then used to obtain explicit formulas for several classes of incomplete games with special structures.

The second part deals with convexity. We begin with considering the case of non-negative minimal incomplete games. Then we survey existing results in the related theory of set functions, namely providing context to the problem of completing partial functions. We provide a characterisation of extendability and a full description of the set of symmetric convex extensions. The set serves as an approximation of the set of convex extensions.

Finally, we outline an entirely new perspective on a connection between incomplete cooperative games and cooperative interval games.
\end{abstract}

\end{frontmatter}


\section{Introduction and motivation}

Developing various approaches to deal with uncertainty is heavily intertwined
with decision making and thus also with game theory. It is only natural since inaccuracy in
data is an everyday problem in real-world situations, be it a lack of
knowledge on the behaviour of others, corrupted data, signal noise or a
prediction of outcomes such as voting or auctions. Since the degree of
applications is so wide, a lot of models of various complexity and
use-case scenarios exist. Among those models most relevant to cooperative game theory, there are fuzzy cooperative games
\cite{Branzei2008,Mares2001,Mares2004}, multi-choice games
\cite{Branzei2008}, cooperative interval games
\cite{Gok2009a,Gok2011,Gok2009b,Bok2015}, fuzzy interval
games~\cite{Mallozzi2011}, games under bubbly uncertainty
\cite{Palanci2014}, ellipsoidal games~\cite{Weber2010},
and games based on grey numbers~\cite{Palanci2015}.

In the theory of classical (transferable utility) cooperative games, we know the precise reward (or payoff) for the cooperation of every group of players, called
\emph{coalition}. In incomplete cooperative games, this is
generally no longer true, since only some of the coalition values are known. This models the uncertainty over data and its
consistency. The model was first introduced in literature by
Willson~\cite{Willson1993} in 1993. Willson gave the basic notion of
incomplete game (called there \emph{partially defined games}) and he generalised the definition of the Shapley value for such
games. Two decades later, Inuiguchi and Masuya revived the research.
In~\cite{Masuya2016a}, they focused mainly on the class of
superadditive games (and also briefly mentioned particular cases of convex
and positive games in which precisely the values of singleton coalitions and
grand coalition are known). Further, Masuya~\cite{masuya2021approximated} considered approximations of the Shapley value for incomplete games, and Yu~\cite{xiaohui2021extension} introduced a generalisation of incomplete games to games with coalition structures and studied the proportional Owen value (which is a generalisation of the Shapley value for these games). Apart from that, Bok and Černý cosidered the property of 1-convexity and related solution concepts (values) for incomplete games~\cite{BC1convex}. 

\subsection{Motivation}

We would like to mention some of the motivations for our reseach.

The first one (already mentioned) is that the model can be seen as one of the possible approaches to uncertainty. This is no doubt a widespread problem in real world and thus in applications. Our paper does not discuss specific scenarios but it provides theoretical foundations for such analysis.

The model of incomplete games also has strong connections to other uncertainty models in cooperative game theory. In particular, we discuss some natural connections to cooperative interval games in Section~\ref{sec:conclusion}. 

The uncertainty issue can be actually turned around. The study of incomplete games shows us what kind of information can we infer if we intentionally forget a part of the input. This is especially valuable since in general the size of cooperative games is exponential in the number of players. Thus we can turn otherwise computationally difficult task into an easier one (in most cases, at some cost of precision).

The \emph{set function} is a function having as its domain the power set of a given set. We note that incomplete cooperative games can also be viewed as \emph{partial set
functions}. Indeed such structures were already
studied, yet without highlighting the relation to the theory of incomplete games.
In particular, extension of partial set functions into so called \emph{submodular functions} was studied e.g.
in~\cite{Bhaskar2018,Bhaskar2019,Seshadhri2014}.
We refer to the exhaustive book of Grabisch~\cite{Grabisch2016}
which discusses in a detail connections of various types of set
functions to entirely different parts of mathematics, with cooperative games being one of them.

A main goal of cooperative game theory is the study of \emph{solution concepts}, those are functions assigning a set of payoff vectors to each game. Solution concepts assiging precisely one payoff vector to each game are called \emph{values}. For symmetric games (also studied in Section~\ref{sec:convexity}), tend to behave in a quite simple way. However, general (multi-point) solution concepts, like \emph{core}, \emph{imputations}, \emph{stable sets}, or \emph{the Weber set} are interesting for symmetric games as they remain to be multi-point even with the symmetry.

We also note that for analysing solution concepts of incomplete games, it is absolutely necessary to first analyse possible sets of extension (belonging to some given class). If a certain set of extensions is difficult to describe, it might be advantageous to describe its proper subsets (e.g. symmetric extensions if the incomplete game is symmetric), which then also serve as \emph{approximations} of the former set. Our results here can be also regarded in this way.

\subsection{Main results and structure of the paper}

Our results concern two important classes of games: \emph{convex games} and
their subclass of \emph{positive games}.
Let us highlight the main contributions of this paper and its structure.

\begin{itemize}
    \item In Section~\ref{sec:preliminaries}, we outline the necessary background of the
    cooperative game theory and also introduce fundamental definitions of
    incomplete cooperative games, both needed further in the text.
    \item In Section~\ref{sec:positive-description}, we study positivity of incomplete games in general. We tackle questions considering extendability to a positive extension, boundedness of the set of positive extensions and provide a description of the set of positive extensions using its extreme games in case the set is bounded. 
    \item Section~\ref{sec:positive-description-spec} is focused on three different classes with special structure of the known values. We analyse these classes as an application of the characterisation of extreme games from the previous section.
    \item Section~\ref{sec:convexity} is dedicated to convexity and to \emph{symmetric convex extensions}. We characterise under which conditions an incomplete game is extendable into a symmetric convex extension. We provide the range of each coalition's worth over all such possible extensions and fully describe the set of symmetric convex extensions as a set of convex combinations of its
    \textit{extreme games}. We also provide a geometrical point of view on the
    set of symmetric convex extensions.
    \item Section~\ref{sec:conclusion} concludes the paper by providing connections between the theory of incomplete cooperative games and cooperative interval games.
\end{itemize}

\section{Preliminaries} \label{sec:preliminaries}

\subsection{Classical cooperative games}

Comprehensive sources on classical cooperative game theory are for
example~\cite{Branzei2008,Driessen1988,Gilles2010,Peleg2007}.
For more on applications, see
e.g.~\cite{Bilbao2012,Curiel2013,Lemaire1991}. Here we present only the
necessary background needed for the study of incomplete cooperative
games. The crucial definition is that of a classical cooperative game. We note that the following definition assumes transferable utility (shortly
TU). 

\begin{definition}
    A \emph{cooperative game} is an ordered pair $(N,v)$ where the set $N=\{1,2,\ldots ,n\}$ and $v\colon 2^N \to \mathbb{R}$ is the characteristic function of the cooperative game. Further, $v(\emptyset) = 0$.
\end{definition}

The set of $n$-person cooperative games is denoted by $\Gamma^n$. The subsets of $N$ are called \emph{coalitions} and $N$ itself is called the
\emph{grand coalition}. We often write $v$ instead of $(N,v)$ whenever there
is no confusion over what the player set is. We often associate the characteristic functions $v\colon 2^N\to\mathbb{R}$ with vectors $v\in\mathbb{R}^{2^{\lvert N \rvert}-1}$.

To avoid cumbersome notation, we use the following abbreviations. We often replace singleton set $\left\{i\right\}$ with $i$. We use $\subseteq$ for the relation of ``being a subset of'' and $\subsetneq$ for the relation ``being a proper subset
of''. By $\emptyset\neq S \subseteq N$, we mean $S \subseteq N$ and $S \neq \emptyset$. To denote the sizes of coalitions e.g. $N,S,T$, we often use $n,s,t$, respectively.

A cooperative game $(N,v)$ is \emph{monotonic} if for every $T \subseteq S \subseteq N$, we have $v(T) \leq v(S)$, and it is \emph{superadditive} if for every $S,T \subseteq N$ such that $S \cap T = \emptyset$, we have $v(S) + v(T) \leq v(S \cup T)$. The set of superadditive $n$-person games is denoted by $S^n$.

\begin{definition}\label{def:convex}
    A cooperative game $(N,v)$ is
    \emph{convex} if for every $S, T \subseteq N$,
        \begin{equation}\label{eq:convex-games}
        v(T) + v(S) \le v(S \cup T) + v(S \cap T).
        \end{equation}
        The set of convex $n$-person games is denoted by $\Co$.
\end{definition}

The property of the characteristic function of a convex game is called \emph{supermodularity}. If Conditions~\eqref{eq:convex-games} hold with opposite inequality, we talk about \emph{submodular functions} (see~\cite{Grabisch2016} for more). Each of the three aforementioned classes incorporates a different approach to formalising the concept of bigger coalitions being stronger. Clearly, convex games are superadditive. The class of convex games is maybe the most prominent class in cooperative game theory since it has many applications and it enjoys both elegant and powerful characterisations. Among them, the following
characterisation by Shapley is necessary for proving our results.

\begin{theorem}\label{thm:convchar} \cite{Shapley1971}
    A cooperative game $(N,v)$ is convex if and only if for every $i \in N$ and every
    $S \subseteq T \subseteq N \setminus \{i\}$, it holds that
    $v(S\cup i) - v(S) \leq v(T\cup i) - v(T)$.
\end{theorem}

\emph{Positive games} (also called \emph{totally monotonic games}) form a subset of convex games. The concept of total monotonicity generalises the notion of monotonicity in cooperative games~\cite{Chateauneuf1989,Fujimoto2005}. The class was introduced by Shapley~\cite{Shapley1953} by using \emph{unanimity games}, special positive games, to study his now well-known value.

\begin{definition}
    For every nonempty $T \subseteq N$, the \emph{unanimity game} $(N,u_T)$ is
    defined as
    $$u_T(S) \coloneqq
    \begin{cases}
        1 & \text{if }T \subseteq S,\\
        0 & \text{otherwise.}\\
    \end{cases}$$
\end{definition}

The set of all unanimity games forms a basis of the vector space of cooperative $n$-person games, i.e.\ any cooperative $n$-person game $v$ can be expressed as $v=\sum_{\emptyset \neq T \subseteq N}d_v(T)u_T$. Coordinates $d_v(T)$ are called \emph{Harsanyi dividends}. Positive games form a non-negative orthant in this coordinate system. In this paper, we employ a different definition of Harsanyi dividends.

\begin{definition}
    For every nonempty $T \subseteq N$, the \emph{Harsanyi dividend} $d_v(T)$
    of a game $(N,v)$ is defined as
    $$d_v(T) \coloneqq \sum_{S \subseteq T}(-1)^{\lvert T \setminus S \rvert}v(S).$$
\end{definition}

\begin{definition}
    A cooperative game $(N,v)$ is \emph{positive} if the Harsanyi
    dividend $d_v(T)$ is non-negative for every nonempty $T \subseteq N$. The set of positive $n$-person games is denoted by $\P$.
\end{definition}

The results concerning positive games are relatively sparse and scattered (the games are also called \emph{totally monotonic games}). We would like to refer the reader to~\cite{Choquet1954,Gilboa1994,Grabisch2000,Rota1964} for further resources.

The dividend $d_v(T)$ has another interesting property; it is also the value
corresponding to $T$ in \emph{Möbius transform} of value $v(T)$
(see~\cite{Grabisch2016} for further details). Another way to
represent $d_v$ is the following. Let
$$d_v(T) = 
\begin{cases}
    0, & \text{if } T = \emptyset,\\
    v(i), & \text{if } T = i \text{ for } i \in N,\\
    v(T) - \sum_{S \subsetneq T}d_v(S), &\text{if } T \subseteq N, \lvert T \rvert > 1.\\
\end{cases}$$

The general case of both positive and convex games
yields very complex situations and so we often restrict ourselves
to simpler subclasses of these games. One of the possible approaches is to impose the additional property of symmetry.

\begin{definition}
    A cooperative game $(N,v)$ is said to be
\emph{symmetric} if for every $S,T \subseteq N$ such that   $\lvert S \rvert =
\lvert T \rvert$, it holds that $v(S) = v(T)$.
\end{definition}

We denote the sets of symmetric convex and symmetric positive $n$-person games by $\Co_\sigma$ and $\P_\sigma$, respectively. It is easy to observe that symmetric games can be described in a succinct way. This helps in our analysis and provides an interesting view-point on our results; while the set of symmetric convex extensions is easy to describe, it is not the case for symmetric positive extensions. We provide examples showing that trying to describe such extensions is a difficult problem to tackle.

\subsection{Incomplete cooperative games}

The following definitions are inspired by~\cite{Masuya2016a}.

\begin{definition}\emph{(Incomplete game)}
    An incomplete game is a tuple $(N,\mathcal K,v)$ where the set $N = \{1,\ldots,n\}$, $\mathcal K \subseteq 2^N$ is the set of coalitions with known values and $v\colon \mathcal K \to \R$ is the characteristic function of the incomplete game. Further, $\emptyset \in \mathcal K$ and $v(\emptyset)=0$.
\end{definition}

Compared to definitions in~\cite{Masuya2016a}, our assumptions are slightly more general. Most importantly, we do not a priori assume that the values of singleton coalitions and the grand coalition are among the known values, i.e.\ that they are in~$\K$.

The fundamental tool to analyse incomplete games are their \emph{$C$-extensions}.

\begin{definition}
    Let $C\subseteq \Gamma^n$ be a class of $n$-person games. A cooperative game $(N,w) \in C$ is a \emph{$C$-extension} of an incomplete game $(N,\K,v)$ if $w(S)=v(S)$ for every $S \in \mathcal K$. 
\end{definition}

The set of all $C$-extensions of an incomplete game $(N,\K,v)$ is denoted by $C(v)$. We write \emph{$C(v)$-extension} whenever we want to emphasize the game $(N,\K,v)$. Also, if there is a $C(v)$-extension, we say $(N,\K, v)$ is \emph{$C$-extendable}. Finally, the set of all $C$-extendable incomplete games with fixed $\K$ is denoted by $C(\K)$.

The sets of $C$-extensions studied in this text are always convex. One of the main goals of the model of incomplete cooperative games is to describe these sets using their extreme points and extreme rays whenever the description is possible. We refer to the extreme points as to \emph{extreme games}. For the sake of completeness, let us recall a formal definition of extreme points. This particular definition is being used later in our proofs.

\begin{definition}\label{def:extreme-points}
    Let $K$ be a convex set. A point $x \in K$ is an \emph{extreme point} (or \emph{vertex}) of $K$ if there is no way to express $x$ as a convex combination $\lambda y + (1-\lambda)z$ such that $y,z \in K$ and $0 \leq \lambda \leq 1$, except for taking $y=z=x$.
\end{definition}

If the structure of $C(v)$ is too difficult to describe and it is bounded from either above or from below, we introduce the lower and the upper game. 

\begin{definition}\label{def:lower-upper-games}
    \emph{(The lower and the upper game of a set of $C$-extensions)} Let
    $(N,\mathcal K,v)$ be a $C$-extendable incomplete game. If $C(v)$ is bounded then the \emph{lower game} $(N,\underline{v})$ and the \emph{upper game} $(N,\overline{v})$ of $C(v)$ are complete games such that for every $(N,w) \in C(v)$ and for every $S \subseteq N$, we have
    \[\underline{v}(S) \leq w(S) \leq \overline{v}(S)\]
    and for every $S \subseteq N$, there are $(N,w_1),(N,w_2) \in C(v)$ such that 
    \[\underline{v}(S) = w_1(S) \text{ and } \overline{v}(S)=w_2(S).\]
\end{definition}

It is important to note that if $C(v)$ is bounded only from above and not from below, the lower game does not exist. The analogous holds for the case of being bounded only from below. These games delimit the area of $\mathbb{R}^{2^n}$ that contains the set of $C$-extensions. Even if we know the description of $C(v)$, the lower and the upper game are still useful as they encapsulate the range of possible profits of coalition $S$ across all possible $C$-extensions by the interval $\left[\underline{v}(S),\overline{v}(S)\right]$.

We remark that it is important to distinguish between lower and upper games of
different sets of extensions. For example, lower games of superadditive and
convex extensions do not coincide in general. There are also examples of
incomplete games where the set of convex extensions might be empty and
therefore, the lower game of convex extensions might not exist.  However, the
same incomplete game can have the lower game of superadditive extensions.

As we already mentioned, we are interested in the property of symmetry in
games. The following is the generalisation of this property to incomplete
games.

\begin{definition}
    An incomplete game $(N,\mathcal{K},v)$ is \emph{symmetric} if for
    every $K_1,K_2 \in \mathcal{K}$ such that $\lvert K_1 \rvert = \lvert K_2
    \rvert$, the equality $v(K_1) = v(K_2)$ holds. 
\end{definition}

\section{Positive extensions}\label{sec:positive-description}

 In~\cite{Masuya2016a}, Masuya and Inuiguchi studied $\P$-extensions of incomplete games with special structure, namely $(N,\K,v)$ with $\K = \{\emptyset, N\} \cup \{\{i\} \mid i \in N\}$ and $v(S) \geq 0$ for $S \in \K$. In our text, we refer to these games as \emph{non-negative minimal incomplete games}. In their work, as a consequence of an approach slightly different from ours, they do not consider the question of $\P$-ex\-ten\-da\-bi\-li\-ty. To characterise $\P$-extendability of non-negative minimal incomplete games, we denote $\Delta \coloneqq v(N) - \sum_{i \in N}v(i)$ and $N_1\coloneqq \{T \subseteq N \mid \lvert T \rvert > 1\}$.
\begin{theorem}
    Let $(N,\K,v)$ be a non-negative minimal incomplete game. It is $\P$-extendable if and only if $\Delta \geq 0$.
\end{theorem}
\begin{proof}
    If $\Delta \geq 0$, it immediately follows that game $(N,w^*)$ defined using its dividends as 
    \[
    d_{w^*}(S) \coloneqq \begin{cases}
        v(i) & \textit{if } S = \{i\},\\
        \Delta & \textit{if } S = N,\\
        0 & \textit{otherwise},\\
    \end{cases}
    \]
    is $\Pv$-extension. If $\Delta < 0$, it follows for any $\Pv$-extension $(N,w)$ that
    \[
    \Delta = \sum_{\emptyset \neq S \subseteq N}d_w(S) - \sum_{i \in N}\delta_w(i) = \sum_{S \in N_1}d_w(S) < 0.
    \]
    As $d_w(S) \geq 0$ for every $S \in N_1$, this leads to a contradiction.
\qed \end{proof}

In~\cite{Masuya2016a}, they showed the lower and the upper game of $\P$-extensions coincide with those of $S^n$-extensions. Games $(N,\underline{v})$ and $(N,\overline{v})$ are defined as
\begin{equation}\label{eq:min-info-lower-upper-games}
\underline{v}(S) \coloneqq \begin{cases}
    v(S), & \cif S \in\K,\\
    \sum_{i \in S}v(i), & \cif S \notin \K,\\
\end{cases}\text{ and }
\overline{v}(S) \coloneqq \begin{cases}
    v(S), & \cif S \in\K,\\
    v(N), & \cif S \notin\K.\\
\end{cases}
\end{equation}

The set of $\P$-extensions of non-negative minimal incomplete games is described in~\cite{Masuya2016a} as a set of all convex combinations of its extreme points. Those correspond to games $(N,v^T)$, parametrised by coalitions $\emptyset \neq T \subseteq N$,
\begin{equation}\label{eq:def-vt}
    v^T(S) = \begin{cases}
        0, &  S = \emptyset,\\
        \Delta + \sum_{i \in S}v(i), &  S \notin\K \cand T \subseteq S,\\
        \sum_{i \in S}v(i), &  S \notin\K\cand T \subsetneq S.\\
    \end{cases}
\end{equation}
Notice similarity between games $(N,v^T)$ and unanimity games $(N,u_T)$.
\begin{theorem}\label{thm:positive-extensions-min-info}\cite{Masuya2016a}
    Let $(N,\K,v)$ be a non-negative minimal incomplete game and let $(N,v^T)$ for $T \in N_1$ be games from (\ref{eq:def-vt}). The set of $\P$-extensions can be expressed as
    \begin{equation}\label{eq:positive-extensions}
        \Pv=\left\{\sum_{T \in N_1}\alpha_Tv^T \mid \sum_{T \in N_1}\alpha_T=1, \alpha_{T} \geq 0\right\}.
    \end{equation}
\end{theorem}

In this section, we generalise the discussed results to general setting. In Subsection~\ref{subsec:positive-extendability}, we provide a characterisation of $\P$-extendability based on duality of linear programming and give an example of its application in a time-complexity analysis of $\P$-extendability of incomplete games with special structures. We also give a sufficient and necessary condition for boundedness of $\P(v)$. In Subsection~\ref{subsec:positive-extreme-games}, we investigate a description of the set of $\Pv$-extensions if the set is bounded. We do so by characterising extreme games of the set by following (and slightly modifying) the proof of the sharp form of Bondareva-Shapley theorem (the theorem was introduced independently by Bondareva in 1963~\cite{Bondareva1963} and Shapley in 1967~\cite{Shapley1967}).
 
\subsection{\texorpdfstring{$\P$}{Pn}-extendability and boundedness of \texorpdfstring{$\P$}{Pn(v)}}\label{subsec:positive-extendability}

To provide a certificate for non-$\P$-extendability of an incomplete game, we employ duality of linear systems. This approach was motivated by a work by Seshadhri and Vondr{\'a}k~\cite{Seshadhri2014} and the so called \emph{path certificate} for non-extendability of submodular functions (corresponding to convex games). Although its size is exponential in the number of players in general, for special cases, the solvability of the dual system is polynomial in $n$ and therefore, the $\P$-extendability is polynomially decidable in $n$ for such cases. In the proof of the characterisation, we use the seminal result of Farkas~\cite{Farkas1902}.

\begin{lemma}\emph{(Farkas' lemma, \cite{Farkas1902})}\label{lem:farkas}
    Let $A \in \mathbb R^{m \times n}$ and $b \in \mathbb R^n$. Then exactly one of the following two statements is true.
    \begin{enumerate}
        \item There exists $x \in \mathbb R^n$ such that $Ax = b$ and $x \geq 0$.
        \item There exists $y \in \mathbb R^m$ such that $A^Ty \geq 0$ and $b^Ty \leq -1$.
    \end{enumerate}
\end{lemma}

\begin{theorem}\label{thm:positive-extendability}
    Let $(N,\K,v)$ be an incomplete game. The game is $\P$-extendable if and only if the following system of linear equations is not solvable:
    \begin{enumerate}
        \item $\forall\,T \subseteq N, T \neq \emptyset: \sum_{S \in \K, T \subseteq S} y(S) \geq 0$,
        \item $\sum_{S \in \K}v(S)y(S) \leq -1$.
    \end{enumerate}
\end{theorem}
\begin{proof}
    Let $M \coloneqq 2^{n}-1$ and $U' \in \mathbb{R}^{M \times M}$
    be a matrix with characteristic vectors of unanimity games $u_T$ as its
    columns. Then it holds that $U'd = w$ for every game $(N,w)$ and its
    vector of Harsanyi dividends $d$.
    
    For an incomplete cooperative game $(N,\K,v)$, we reduce matrix
    $U'$ by deleting the rows corresponding to coalitions with unknown values,
    reaching a system $Ud=v$. This adjustment eliminates unknowns on the right
    hand side of the equation, yet no information about the complete game is
    lost since the vector of Harsanyi dividends carries full information. 
    
    Game $(N,\K,v)$ is $\P$-extendable if and only if $Ud=v$ is solvable for $d \geq0$. By Farkas' lemma (Lemma~\ref{lem:farkas}) this happens if and only if
    the following system has no solution,
    \begin{equation}\label{eq:farkas}
        U^Ty \geq 0 \text{ and } v^Ty \leq -1.
    \end{equation}
    The conditions given by~\eqref{eq:farkas} correspond to those from
    the statement of the theorem.
\qed \end{proof}

Notice that even though the number of inequalities $\sum_{S \in K, T \subseteq
    S} y(S) \geq 0$ is $2^{n}-1$ (since we have one inequality for
every $\emptyset \neq T \subseteq N$), the actual number of \textit{distinct}
inequalities is not larger than $2^{\lvert \K \rvert}-1$ because
each inequality sums over a subset of $\K$.  Depending on the
structure of $\K$, the actual number might be even smaller as is shown in
the following result.

\begin{theorem}\label{prop:positive-extend-example}
    Let $(N,\K,v)$ be an incomplete game such that sizes of all $S \in \K$ are bounded by a fixed constant $c$. Then the problem of $\P$-extendability is po\-ly\-no\-mially-time solvable in $n$. 
\end{theorem}
\begin{proof}
    In $(N,\K,v)$, the number of coalitions with a defined value is at most $\sum_{i = 1}^{c}{n\choose i}$,
    which is a polynomial in $n$. Also, if we consider the linear system from
    Theorem~\ref{thm:positive-extendability}, every $T \subseteq N$ such that $\lvert
    T \rvert > c$ yields an empty sum in its corresponding inequality. Therefore,
    the number of unique conditions in the problem is bounded by the number
    of coalitions with defined value, that is by the sum $\sum_{i =
        1}^{c}{n\choose i}$. We conclude that the linear system can be solved in
    polynomial time by means of linear programming.
\qed \end{proof}

Now we address the question of boundedness of $\Pv$. Notice, the set of $\P$-extensions is always bounded from below, as for every $\P$-extension $(N,w)$, $w(S) = \sum_{\emptyset\neq T \subseteq S}d_w(T)$ and $d_w(T) \geq 0$ for every $\emptyset \neq T \subseteq N$. Therefore, $0$ serves as a lower bound the value of any coalition $S \subseteq N$. To find the lower bound that is binding the profit of every coalition as well as the binding upper bound (hence the lower and the upper game) remains an open problem.

\begin{theorem}\label{prop:positive-bound}
    Let $(N,\K,v)$ be a $\P$-extendable incomplete game. The set of positive extensions $\P(v)$ is bounded if and only if $N \in \K$.
\end{theorem}
\begin{proof}
    If $N \in \K$, then for any $\Pv$-extension $(N,w)$, $\sum_{T
        \subseteq N}d_w(T) = v(N)$, and since for all $\emptyset \neq T \subseteq N$, $d_w(T) \geq
    0$, it follows $d_w(T) \in \left[0,v(N)\right]$. This yields a bound
    (possibly an overestimation) for all possible values of $d_w(T)$. Since the
    dividends are bounded, the set $\P(v)$ is also bounded.
    
    If $N \notin \K$, then the value of coalition $N$ can be arbitrarily large, since there is no upper bound on $d_w(N)$ for a $\P(v)$-extension $(N,w)$. Thus, $\P(v)$ is not bounded.
\qed \end{proof}

\subsection{Description of the set of \texorpdfstring{$\P$}{Pn}-extensions}\label{subsec:positive-extreme-games}

For an incomplete game $(N,\K,v)$, the set of $\Pv$-extensions can be described as
\[
\P(v) = \bigg\{(N,w) \Big|\, \forall S \in \K: w(S) = v(S) \text{ and } \forall T\subseteq N: d_w(T) \geq 0\bigg\},
\]
or equivalently in terms of dividends and $M \coloneqq 2^{n}-1$ as
\[
\P_d(v) \coloneqq \bigg\{ d_w \in \mathbb{R}^{M}\Big|\, \forall S \in \K: \sum_{T \subseteq S} d_w(T) = v(S),\forall T \subseteq N: d_w(T) \geq 0\bigg\}.
\]
Notice that $\P(v) \neq \P_d(v)$ since the former is a set of cooperative games and the latter is a set of vectors of dividends.

Both sets are formed by intersections of closed half-space, thus they are (convex) polyhedrons. If we suppose that $(N,\K,v)$ is $\P$-extendable, then both sets are nonempty. Furthermore, the sets are bounded if and only if $N \in \K$. Bounded convex polyhedrons are convex hulls of their extreme points.

To be able to freely neglect the distinction between extreme points of both sets, we recall a basic result from linear algebra. For a sake of completeness, we include the proof.

\begin{lemma}\label{lem:extreme-image}
    Let $P$ be a convex subset of $\mathbb{R}^{n}$, $A \in \mathbb{R}^{n \times
        n}$ a nonsingular matrix, and $x \in P$ an extreme point of $P$. Then $Ax$ is
    an extreme point of the convex set $A(P) \coloneqq \left\{Au \lvert u \in P\right\}$.
\end{lemma}
\begin{proof}
    Suppose that $x \in P$ is an extreme point of $P$ and the image $Ax$ is not an
    extreme point of $A(P)$. Therefore, there are $Au,Av \in A(P)$ and
    $\alpha \in (0,1)$ such that $\alpha Au + (1 - \alpha) Av = Ax$. But then
    $\alpha Au + (1 - \alpha) Av = A (\alpha u + (1 - \alpha)v)=Ax$, and
    therefore, $x$ is not an extreme point of $P$, as it is a nontrivial convex combination of $u,v \in P$. This is a contradiction.
\qed \end{proof}

Let $U \in \mathbb R^{M\times M}$ be a matrix with vectors of unanimity games $u_T \in \mathbb{R}^{M}$ as columns. It holds that $Ud_w = w$ where $w \in \mathbb{R}^{M}$ is a characteristic vector of game $(N,w)$ and $d_w \in \mathbb{R}^{M}$ represents a vector of Harsanyi dividends of the game. Since unanimity games form a basis of $\mathbb R^M$, the matrix $U$ is nonsingular and thus, by Lemma~\ref{lem:extreme-image}, the extreme points of $\P(v)$ correspond to those of $\P_d(v)$, allowing us to further consider those instead of the former ones.

Following the proof of the sharp form of Bondareva-Shapley theorem from~\cite{Peleg2007}, we give an insight into the description of extreme games of $\P(v)$. We show that for these games, the set of coalitions zero dividends is inclusion-wise maximal.

Our result is based on the following characterisation of extreme points of polyhedrons.

\begin{lemma}\cite{Peleg2007}\label{lem:pos_extreme_points}
    Let $P$ be a polyhedron given by
    \[P \coloneqq \left\{x \in \mathbb{R}^k \big\lvert\, \sum_{j=1}^k a_{tj}x_j \geq b_t, t = 1,\dots, m\right\}.\]
    For $x \in P$, let $S(x) \coloneqq \big\{t \in \{1,\dots,m\}\lvert\, \sum_{j=1}^ka_{tj}x_j=b_t\big\}$. The point $x \in P$ is an extreme point of $P$ if and only if the system of linear equations 
    \[
    \sum_{j=1}^ka_{ij}y_j = b_t \text{ for all } t \in S(x) 
    \]
    has $x$ as its unique solution.
\end{lemma}

Applying Lemma~\ref{lem:pos_extreme_points}, $d_e \in \mathbb{R}^M$ is an extreme game of $\P_d(v)$ if and only if there is no $d_w\neq d_e$ such that $d_w(T) = 0 \iff d_e(T)=0$ for every nonempty $T \subseteq N$. For any $\P(v)$-extension $(N,w)$, we denote by $E(w)$ the \textit{set of negligible coalitions} defined as $E(w)\coloneqq\{T \subseteq N \mid d_w(T)=0\}$. This set proves itself useful in the following lemma. The lemma states that inclusion-maximality of $E(e)$ across $E(x)$ for $d_x \in \P_d(v)$ is equivalent with uniqueness of $E(e)$ across $E(x)$ for $d_x \in \P_d(v)$. Together with Lemma~\ref{lem:pos_extreme_points}, this connects the extremality of games with the inclusion-ma\-xi\-ma\-li\-ty of sets $E(e)$.
\begin{lemma}\label{lem:max-unique-char}
    Let $(N,\K,v)$ be a $\P$-extendable incomplete game and $d_e \in \P_d(v)$. Then the following are equivalent:
    \begin{enumerate}
        \item there is no $d_x\in \P_d(v)$ such that $E(e) \subsetneq E(x)$,
        \item there is no $d_y\in\P_d(v)$ different from $d_e$, such that $E(e) = E(y)$.
    \end{enumerate}
\end{lemma}
\begin{proof}
First, suppose that there is $d_x \in \P_d(v)$ such that $E(e) \subsetneq E(x)$. We show that there is not only one, but infinitely many vectors $d_{y^\alpha} \in P_d(v)$ different from $d_e$ such that $E(e)=E(y^\alpha)$. The idea is to take any non-trivial convex combination $d_{y^\alpha}\coloneqq\alpha d_e + (1-\alpha)d_x$ for $0 < \alpha < 1$. Such game is clearly positive (a convex combination of non-negative dividends remains non-negative) as it is also an extension of $(N,\K,v)$, because for every $S \in \K$,
\[
\sum_{T \subseteq S} d_{y^\alpha}(T) = \alpha \sum_{T \subseteq S} d_{e}(T) + (1-\alpha) \sum_{T \subseteq S} d_{x}(T) = \alpha v(S) + (1-\alpha)v(S) = v(S).
\]
And since $d_x \neq d_e$, there is $S \notin\K$ such that $x(S)\neq e(S)$ for which
\[
y^\alpha(S) = \sum_{T \subseteq S}d_{y^\alpha}(T) = \alpha \sum_{T \subseteq S}d_{x}(T)  + (1-\alpha) \sum_{T \subseteq S}d_{e}(T) = \alpha x(S) + (1-\alpha)e(S).
\]
Therefore, any two parameters $\alpha_1,\alpha_2$ such that $0 < \alpha_1 < \alpha_2 < 1$ yield different values $y^{\alpha_1}(S) \neq y^{\alpha_2}(S)$, thus $d_{y^{\alpha_1}}\neq d_{y^{\alpha_2}}$.

Now suppose that there is $d_y \in \P_d(v)$ different from $d_e$ such that $E(e)=E(y)$. We take a combination $d_z=d_e - \beta (d_y - d_e)$ with $\beta$ such that for at least one $S \notin E(e)$, $d_z(S)=0$. Thus $E(e) \subseteq E(z)$ and still, $d_z \in \P_d(v)$. For such $S$, it must hold
\[
d_z(S) = d_e(S) - \beta (d_y(S) - d_e(S)) = 0,
\]
therefore $\beta = \frac{d_e(S)}{d_y(S)-d_e(S)}$. We have to choose $S$ such that $d_y(S)\neq d_e(S)$. Furthermore, we have to secure that for every $T \notin E(e)$, $d_z(T) \geq 0$, or equivalently
\begin{align*} 
d_z(T) &= d_e(T) - \beta (d_y(T) - d_e(T))\\
&= d_e(T) - \frac{d_e(S)}{d_y(S)-d_e(S)} (d_y(T) - d_e(T)) \geq 0
\end{align*}
This can be done by taking the minimum for $S$ over all such coalitions $T$, i.e.\ 
\[\beta \coloneqq \min\limits_{T \notin E(e): d_e(T)\neq d_y(T)}\frac{d_e(T)}{d_y(T)-d_e(T)}.\] 
Then for $T \notin E(e)$, $d_z(T)\geq0$, since it is equal to
\[
d_e(T) - \frac{d_e(S)}{d_y(S)-d_e(S)} (d_y(T) - d_e(T)) \geq d_e(T) - \frac{d_e(T)}{d_y(T)-d_e(T)} (d_y(T) - d_e(T)).
\] 
Clearly, the last expression is equal to zero. Finally, for $K \in \K$,
\[
z(K) = \sum_{C \subseteq K}d_z(K) = \sum_{C \subseteq K}d_e(K) - \beta \left(\sum_{C \subseteq K} d_y(K) - \sum_{C \subseteq K}d_e(K)\right),
\]
and since all the three sums in the last expression are equal to $v(K)$, we conclude that $z(K)=v(K)$ and thus, $d_z \in \P_d(v)$.
\qed \end{proof}

The following characterisation of extreme points follows as a direct application of Lemma~\ref{lem:pos_extreme_points} and~\ref{lem:max-unique-char}.

\begin{theorem}\label{thm:vertex-positive}
    For $\P$-extendable incomplete game $(N,\K,v)$, it holds $(N,e)$ is an extreme game of $\Pv$ if and only if its set of negligible coalitions $E(e)$ is inclusion-maximal, i.e.\ there is no $(N,w)\in \P(v)$ such that $E(e) \subsetneq E(w)$.
\end{theorem}

\section{Application to analysis of positive extensions for special cases}\label{sec:positive-description-spec}

This section contains an analysis of $\P$-extensions of several classes of incomplete games. 
We show a direct application of Theorem~\ref{thm:vertex-positive} to the description of the set of $\P$-extensions for three classes of incomplete games. We do not show only a derivation of extreme games but also a derivation of the lower and the upper game together with a characterisation of $\P$-extendability.

\subsection{Pairwise disjoint coalitions of known values}
For the first class of incomplete games it holds that the coalitions with known values (excluding $N$) are pairwise-disjoint.

\begin{theorem}\label{thm:empty-int}
    Let $(N,\K,v)$ be a $\P$-extendable incomplete game, where $\K = \left\{S_1,\dots,S_{k-1},N\right\}$ and for all $i,j \in \{1,\ldots,k-1 \}$,
    it holds that $S_i \cap S_j = \emptyset$. Then the extreme games
    $v^{\mathcal{T}}$, the lower game $\underline{v}$, and the upper game $\overline{v}$ can be described as follows:
    $$v^{\mathcal{T}}(S)\coloneqq\begin{cases}
        0, & \text{if } \nexists T \in \K: T \subseteq S,\\
        \sum_{i:T_i \subseteq S}v(S_i), & \text{if } \exists T \in \K: T \subseteq S \text{ and } T_N \nsubseteq S\\
        v(N) - \sum_{i: T_i \nsubseteq S}v(S_i), & \text{if } \exists T \in \K: T \subseteq S \text{ and } T_N \subseteq S,\\
    \end{cases}$$
    \[
    \underline{v}(S) \coloneqq v^{\K}(S)=\begin{cases}
        0, & \text{if } \nexists T \in \K: T \subseteq S,\\
        \sum_{i:S_i \subseteq S}v(S_i), & \text{if } \exists T \in \K: T \subseteq S \text{ and } N \neq S,\\
        v(N), & \text{if } \exists T \in \K: T \subseteq S \text{ and } N = S,\\
    \end{cases}
    \]
    \[
    \overline{v}(S) \coloneqq
    \begin{cases}
        v(S_i), & \text{if } S \subseteq S_i,\\
        v(N) - \sum_{i: S_i \nsubseteq S}v(S_i), & \text{otherwise},\\
    \end{cases}
    \]
    where $\mathcal{T} \coloneqq \left\{T_1,\dots,T_{k-1},T_N\right\}$ such that $T_i
    \subseteq S_i$, $T_N \subseteq N$ and $T_N \nsubseteq S_\ell$ for any $\ell
    \in \{1,\ldots,k-1\}$.\\
    Furthermore, the $\P$-extendability of $(N,\K,v)$ is characterised by a condition \[v(N) \geq \sum_{i=1}^{k-1} v(S_i).\]
\end{theorem}
\begin{proof}
    Let $(N,\K,v)$ be an incomplete game with the properties above. For
    any $\P(v)$-extension $(N,w)$, from the fact that the coalitions in
    $\K \setminus \left\{N\right\}$ are disjoint, at least one
    subcoalition $T_i$ of each coalition $S_i \in \K \setminus
    \left\{N\right\}$ must have a nonzero dividend $d_w(T_i)$, otherwise $v(S_i) =
    0$. By Theorem~\ref{thm:vertex-positive}, there is at most one such subcoalition if we consider an extreme game. If there were two nonzero dividends $d_w(T_i^1)$, $d_w(T_i^2)$ for one $S_i$, then the corresponding set of negligible coalitions would not be maximal. Setting the dividend of $T_i^1$ to $d_{w}(T_i^1) + d_{w}(T_i^2)$ yields a set $E$, such that $E(w) \subsetneq E$. By this, for the extreme game, it holds $d_{w}(T_i) = v(S_i)$. We further see, since $v(N) = \sum_{T \subseteq N} d_{w}(T)$, that
    $v(N) \geq \sum_{S_i \in \K \setminus \left\{N\right\}}v(S_i)$ holds.
    If the inequality does not hold, then that there is no extreme game of $\P(v)$ and hence, since the set is bounded ($N \in \K$), it is not  $\P$-extendable. Now, if the inequality is strict, there has to be another nonzero
    dividend of a coalition $T_N \subseteq N$ such that $T_N \nsubseteq S_i$ for
    $S_i \in \K \setminus \left\{N\right\}$, otherwise $T_i, T_N$ are two
    distinct subsets of $S_i$ and $E(w)$ is not maximal.
    Again, since we are interested in extreme games, by
    Theorem~\ref{thm:vertex-positive}, there is only one such coalition $T_N$,
    resulting in $d_{w}(T_N) = v(N) - \sum_{S_i \in \K \setminus
        \left\{N\right\}}v(S_i)$. Any game parameterised by a collection $\mathcal{T} \coloneqq
    \left\{T_1,\dots,T_{k-1},T_N\right\}$ and expressed as $v^\mathcal{T}$ from
    the statement of the theorem is thus an extreme game of $\P(v)$.
    
    Now let us show that the game $v^\K$ is the lower game.
    For a coalition $S$ with no subcoalition contained in $\K$,
    $v^\K(S)=0=\underline{v}(S)$. For a coalition $S$ such that there is $T \in K$, $T \subseteq S$, the value of $w(S)$ of any $\P(v)$-extension cannot be smaller than the sum $\sum_{T: T \in \K, T \subseteq    S}v(T) = v^\K(S)$. And since $N \in \K$, $v^{\mathcal{T}}(N) = v(N) = \underline{v}(N)$.
    
    Finally, we show that each value of the upper game is achieved by a different
    extreme game. If $S$ is a proper subcoalition of $S_i$, the value $v(S_i)$
    is, thanks to the non-negativity of dividends, an upper bound for the value
    of $S$. For any extreme game $v^{\mathcal{T}}$ such that $S \in \mathcal{T}$,
    this bound is tight. If $S$ is not a subcoalition of any $S_i$, its value
    cannot exceed $v(N) - \sum_{T_i\in \mathcal{T}\setminus{T_N}}d_{w}(T_i)$,
    otherwise the characterisation of $\P$-extendability is not satisfied for the
    grand coalition $N$. By taking an extreme game with $T_N = S$, we see that
    this bound is tight.
\qed \end{proof}

\subsection{Set of known values $\K$ closed on subsets}
The second class of incomplete games satisfies that the set $\K \setminus \{N\}$ is closed on subsets, i.e.\ $S \in K, T \subseteq S \implies T \in\K$. The analysis of this case helps us further in the study of symmetric positive extensions ($\P_\sigma$-extensions).

\begin{theorem}\label{thm:case-down-closed}
    Let $(N,\K,v)$ be a $\P$-extendable incomplete game such that $N \in \K$ and for every $S\in \K \setminus \left\{N\right\}, T \subseteq S \implies T \in\K$. Furthermore, for $S \in \K$, let $\delta_S$ be defined as $\delta_{\{i\}} = v(\{i\})$ and $\delta_S = v(S) -
    \sum_{T \subsetneq S}\delta_T$. Then the extreme games $v^{C}$, the lower game $\underline{v}$, and the upper game $\overline{v}$ can be described as follows:
    \[
    v^{C}(S)\coloneqq\begin{cases}
        \delta_N + \sum_{T \in \K, T \subseteq S} \delta_T, & \text{if } C \subseteq S,\\
        \sum_{T \in \K, T \subseteq S} \delta_T, & \text{otherwise,}\\
    \end{cases}
    \]
    for $C \notin \K \setminus \left\{N\right\}$, and 
    \[
    \underline{v}(S) \coloneqq v^{N}(S) =\begin{cases}
        \delta_N + \sum_{T \in \K, T \subseteq S} \delta_T, & \text{if } S = N,\\
        \sum_{T \in \K, T \subseteq S} \delta_T, & \text{otherwise,}\\
    \end{cases}
    \]
    \[
    \overline{v}(S) \coloneqq \begin{cases}
        v(S), & \text{if } S \in \K,\\
        v^{S}(S), & \text{otherwise.}\\
    \end{cases}
    \]
    Furthermore, $(N,\K,v)$ is $\P$-extendable if and only if $\delta_S \geq 0$ for all $S \in \K$.
\end{theorem}
\begin{proof}
    Let $(N,w) \in \P(v)$. Thanks to the structure of $\K$, the dividends $d_w(S)$ for $S \in \K \setminus \left\{N\right\}$ are the same for any $(N,w) \P(v)$and they are equal to $\delta_S$. As a consequence, for any $S$ such that $\delta_S = 0$ it holds $S \in E(w)$ and this holds for any $\P(v)$-extension. Now if the uniquely defined value $\delta_N=v(N) - \sum_{S \in \mathcal{K} \setminus \left\{N\right\}}\delta_S > 0$, there has to be at least
    one $C \notin \K \setminus \left\{N\right\}$ such that its dividend
    $d_w(C)\neq0$. By Theorem~\ref{thm:vertex-positive}, following a similar    argument as in the proof of the previous theorem, $E(w)$ is maximal if and only if there is only one such $C$, otherwise if there are $C_1\neq C_2$ such that $d_w(C_1)\neq 0$ and $d_w(C_2)\neq 0$, by taking $(N,x) \in \P(v)$ such that $d_x(C_1)=0$, $d_x(C_2)=d_w(C_1) + d_w(C_2)$ we arrive into contradiction with maximality, since $E(w) \subsetneq E(x)$. Thus choosing $(N,w)\in \P(v)$, such that $d_w(C)=\delta_n$ yields an extreme game $v^C$ of
    $\P(v)$ for any $C \notin \K \setminus \{N\}$.
    
    For any coalition $S$, its value in any $\P(v)$-extension has to be larger
    or equal to $\sum_{T \in \K, T \subseteq S}\delta_S$. Notice that
    $v^{N}(S)$ is equal to this number for any $S$, thus being the lower game.
    
    For any coalition $S$, its maximal value is either $v(S)$ if $S \in \mathcal
    K$, or at most $v(N) - \sum_{T \in \K \setminus \left\{N\right\}: T
        \nsubseteq S}\delta_T = \delta_N + \sum_{T \in \K, T \subseteq S}\delta_T$, which is
    equal to $v^S(S)$ and thus it is the upper game.
\qed \end{proof}

For both studied classes of incomplete games, it holds $\underline{v}(S)\in\P(v)$. Also notice that the number of extreme games $v^\mathcal{C}$ equals the number of coalitions $C$ such that $C \notin K \setminus \{N\}$, that is $2^n-\lvert K \rvert + 1$ if $v(N) - \sum_{S \in \K \setminus \left\{N\right\}}\delta_S > 0$, otherwise $\P(v)$ contains precisely one game (in case $v(N) - \sum_{S \in \K \setminus \left\{N\right\}}\delta_S = 0$) or no game at all (if $v(N) - \sum_{S \in \K \setminus \left\{N\right\}}\delta_S < 0$).

\subsection{Symmetric positive extensions}

We denote the set of symmetric positive extensions of $(N,\K,v)$ by $\P_\sigma(v)$. Analogously to study of $\Co_\sigma(v)$, we make use of the \textit{reduced forms} $(N,s)$ and $(N,\mathcal{X},\sigma)$ of games $(N,v)$ and $(N,\K, v)$, respectively, which are defined in Definition~\ref{def:reduced-forms}. We can easily obtain the following result as a corollary of Theorem~\ref{thm:case-down-closed}.

\begin{theorem}
    Let $(N,\mathcal{X}, s)$ be the reduced form of a symmetric incomplete game such that $n \in\mathcal{X}$ and $ i \in N, i \leq k \implies i \in \mathcal{X}$. Then the lower game and the upper game of $\P_\sigma(v)$ can be described as
    \[
    \underline{s}(i) \coloneqq \begin{cases}
        s(i), & \text{for } i \in \mathcal{X},\\
        s(k), & \text{otherwise,}\\
    \end{cases}
    \text{\quad and\quad}
    \overline{s}(i) \coloneqq \begin{cases}
        s(i), & \text{for } i \in \K,\\
        s(n), & \text{otherwise.} 
    \end{cases}     
    \]
\end{theorem}

The following game illustrates that even in the symmetric scenario, there is $(N,\mathcal{X},\sigma)$ such that $(N,\underline{s}) \not\in\P_\sigma(v)$.

\begin{example}\emph{(The lower game is not necessarily a $P^4_\sigma$-extension)}
    Let $(N,\mathcal X,\sigma)$ be the reduced form of a symmetric $4$-person incomplete game such that $\mathcal X = \left\{2,4\right\}$. From the properties of symmetric positive games we know that any $(N,s) \in P^4_\sigma(v)$ is given by 4 non-negative dividends with corresponding values $d_1,d_2,d_3,d_4$ such that
    \begin{itemize}
        \item $s(1) = d_1$,
        \item $s(2) = 2d_1 + d_2$,
        \item $s(3) = d_3 + 3d_2 + 3d_1$,
        \item $s(4) = d_4 + 4d_3 + 6d_2 + 4d_1$.
    \end{itemize}
    By setting $d_1 \coloneqq 0, d_2 \coloneqq \sigma(2),d_3 \coloneqq 0$, and $d_4 \coloneqq \sigma(4) - 6d_2$ we
    get a $P^4_\sigma$-extension where $s(1) = 0$ (clearly the minimum) and it is achieved if and only if $d_1=0$. Setting $d_1=0$ yields $s(3) = 3\sigma(2)$. However, to minimize
    $s(3)$, we can choose $d_1 \coloneqq \frac{\sigma(2)}{2}$, $d_2 \coloneqq 0$, $d_3 \coloneqq 0$, and
    $d_4 \coloneqq \sigma(4) - 4d_1$, obtaining $s(3) = 3d_1 = \frac{3}{2}\sigma(2)$. We cannot minimize both values simultaneously and thus $(N,\underline{s})\notin \P_\sigma(v)$.
\end{example}

It is not difficult to generalise this example for symmetric $n$-person games. For similar reasons, even the lower game of (non-symmetric) $\P$-extensions of non-symmetric incomplete games is not
contained in $\P(v)$. This is contrary to what we showed for the classes of incomplete games in Theorem~\ref{thm:empty-int} and~\ref{thm:case-down-closed}.

\section{Convex extensions} \label{sec:convexity}

For non-negative incomplete games with minimal information, sets of $S^n$-ex\-ten\-si\-ons and $\P$-extensions are described in~\cite{Masuya2016a}. For the sake of completeness, in this section we derive similar results for the set of $\Co$-extensions.
\begin{theorem}
    Let $(N,\K,v)$ be a non-negative incomplete game with minimal information. It is $\Co$-extendable if and only if $\Delta \geq 0$.
\end{theorem}
\begin{proof}
    If $\Delta \geq 0$, it immediately follows that game $(N,w^*)$ defined using its dividends as 
    \[
    d_{w^*}(S) \coloneqq \begin{cases}
        v(i) & \textit{if } S = \{i\},\\
        \Delta & \textit{if } S = N,\\
        0 & \textit{otherwise},\\
    \end{cases}
    \]
    is $\Co$-extension. If $\Delta < 0$, it follows $v(N)<\sum_{i \in N}v(i)$, thus any extension of $(N,\K,v)$ cannot be convex.
\qed \end{proof}

In~\cite{Masuya2016a}, they showed the lower and the upper game of $\P$-extensions coincide with those of $S^n$-extensions, thus they must coincide with the lower and the upper game of $\Co$-extensions as well (see~\eqref{eq:min-info-lower-upper-games}). Finally, we derive a description of the set of $\Co$-extensions. We employ $N_1\coloneqq \{T \subseteq N \mid \lvert T \rvert > 1\}$.
\begin{theorem}\label{thm:convex-extensions-min-info}
    Let $(N,\K,v)$ be a non-negative incomplete game with minimal information, and let $(N,v^T)$ for $T \in N_1$ be games from (\ref{eq:def-vt}). The set of $\Co$-ex\-ten\-si\-on can be expressed as
    \begin{equation}\label{eq:convex-extensions}
        \Cv=\left\{\sum_{T \in N_1}\alpha_Tv^T \mid \sum_{T \in N_1}\alpha_T=1, \forall S_1,S_2 \subseteq N: \sum_{T \in E(S_1,S_2)}\alpha_{T} \geq 0\right\},
    \end{equation}
    where $E(S_1,S_2)\coloneqq  \{T \subseteq S_1 \cup S_2 \,|\, T \nsubseteq S_1 \text{ and } T \nsubseteq S_2\}$.
\end{theorem}
\begin{proof}
    The proof follows from the proof of Theorem 6 in~\cite{Masuya2016a}. The only difference is in the condition for coefficients $\alpha_T$. For the description of the set of $S^n$-extensions, a condition $\sum_{T \in E(S_1,S_2)}\alpha_{T} \geq 0$ for every pair of conditions $S_1\cap S_2 =\emptyset$ is enforced. This condition corresponds to the fact that for $S_1,S_2 \subseteq N$ such that $S_1 \cap S_2 = \emptyset$, it holds $v(S_1) + v(S_2) \leq v(S_1 \cup S_2)$. In terms of Harsanyi dividends, it is equivalent to $\sum_{T \in E(S_1,S_2)}\delta_v(T) \geq 0$. For convex games and $S_1,S_2 \subseteq N$ (not necessarily disjoint coalitions), the conditions $v(S_1) + v(S_2) \leq v(S_1 \cap S_2) + v(S_1 \cup S_2)$ can be equivalently expressed in terms of Harsanyi dividends as
    \[
    \sum_{T \subseteq S_1 \cup S_2, T \nsubseteq S_1, T \nsubseteq S_2}\alpha_T \geq 0.
    \]
    Notice that coalitions $T$ are exactly those from the set $E(S_1,S_2)$.
\qed \end{proof}

In our attempt, to derive similar results for a more general setting, we surveyed existing results regarding submodular set functions (recall a set function $v \colon 2^N \to \mathbb{R}$ is submodular if and only if $-v$ is supermodular).

The study of extendability of submodular functions initiated Seshadhri
and Von\-dr{\'a}k in~\cite{Seshadhri2014}. They introduced
\textit{path certificate}, a combinatorial structure whose existence certifies
that a submodular function is not extendable. They also showed an example of a
partial function defined on almost all coalitions that is not extendable, but
by removing a value for any coalition, the game becomes extendable. Later in
2018, Bhaskar and Kumar~\cite{Bhaskar2018} studied extendability of several classes of set functions, including submodular functions. Inspired by the results of Seshadhri and Vondr{\'a}k, they introduced a more natural combinatorial certificate of non-extendability --- \textit{square certificate}. Using this concept, they were able to show that a submodular function is extendable on the entire domain if and only if it is extendable on the lattice closure of the sets with defined values. The \emph{lattice closure} $LC(\K)$ of a set of points $\K \subseteq 2^N$ in a partially ordered set $(2^N,\subseteq)$ is the inclusion-minimal subset of $2^N$ that contains $\K$ and that is closed under the operation of union and intersection of sets. Following is a modification of Theorem 7 from~\cite{Bhaskar2018}.
\begin{theorem}\cite{Bhaskar2018} \label{thm:BK2018}
    Let $(N,\K,v)$ be an incomplete cooperative game and 
    \[\mathcal{F} \coloneqq LC(\K) \cap \{S \subseteq N \mid \underline{S},\overline{S} \in\K \text{ s.t. } \underline{S} \subseteq S \subseteq \overline{S}\}.\] Incomplete game $(N,\K,v)$ is $\Co$-extendable if and only if there is supermodular $w\colon 2^\mathcal{F} \to \mathbb{R}$ such that $w(S)=v(S)$ for $S \in \K$.
\end{theorem}

In 2019, the same authors showed that the problem of
extendability for a subclass of submodular functions, so called
\textit{coverage functions} (see~\cite{Bhaskar2019}.), is NP-complete. Thus, the question of $\Co$-extendability is in general NP-complete as well.

The rest of questions concerning $\Co$-extensions of general incomplete games remain open problems. From now on, we focus on extensions that are both convex and symmetric. The reasons are twofold. First, symmetry yields a simpler analysis of the set of $\Co$-extensions. Second, symmetric $\Co$-extensions form an important subset of $\Co$-extensions and may be considered as an approximation of the set.

\subsection{Symmetric convex extensions}

Since the set of convex extensions seems to be quite hard to describe in its
full generality, we focus on a subset of $\Co$-extensions that are symmetric and denote this set by $\Co_\sigma$. The additional property of symmetry yields compact (and by our opinion elegant) descriptions of the set of $\Co_\sigma$-extensions and since $\Co_\sigma(v) \subseteq \Co(v)$ for symmetric incomplete games, the set of symmetric convex extensions may be regarded as an approximation of set $\Co(v)$.

The main ingredient for our results is the following cha\-ra\-cte\-ri\-sa\-ti\-on of
symmetric convex games. For completeness, the proof of this folklore result is provided in appendix.

\begin{proposition}\label{prop:symconvchar}
    Let $(N,v)$ be a symmetric cooperative game. Then for every $S \subsetneq N \setminus j$ and $i \in S$, it holds that
    \begin{equation}\label{eq:convex-symmetric}
        v(S) \leq \frac{v(S\setminus i) + v(S\cup j)}{2}
    \end{equation}
    if and only if the game is convex.
\end{proposition}

We note that the characterisation from Proposition~\ref{prop:symconvchar} does not hold for general convex games. This can be seen in the following example.

\begin{example}\label{example}
    \emph{(A convex game not satisfying conditions from Proposition~\ref{prop:symconvchar})}\\
    The game $(N,v)$ given in Table~\ref{tab:example} is convex, as can be easily checked. However, the inequality 
    \[
    v(\{1,3\}) \leq \frac{v(\{1\}) + v(\{1,2,3\})}{2}
    \] 
    is not satisfied, as $6 \nleq \frac{1 + 9}{2}$.
    \begin{table}\label{tab:example}
        \begin{center}
            \begin{tabular}{ |c|c|c|c|c|c|c|c| } 
                \hline
                $S$& $\{1\}$ & $\{2\}$ & $\{3\}$ & $\{1,2\}$ & $\{1,3\}$ & $\{2,3\}$ & $\{1,2,3\}$ \\ 
                \hline
                $v(S)$ & $1$ & $1$ & $1$ & $4$ & $6$ & $4$ & $9$ \\ 
                \hline
            \end{tabular}
            \caption{The game $(N,v)$ from Example~\ref{example} with its characteristic function given in the table.}
        \end{center}
    \end{table}
\end{example}   

For symmetric games, we can denote by $s(k)$ the value of $v(S)$ of any $S
\subseteq N$ such that $\lvert S \rvert = k$. This allows us to
formulate the following characterisation of symmetric convex games.

\begin{theorem}\label{thm:symconv}
    A game $(N,v)$ is symmetric convex if and only if for all $k \in \{1, \dots, n-1\},$
    \begin{equation}\label{eq:convex-sym-char}
        s(k)\leq \frac{s(k-1)+s(k+1)}{2}.
    \end{equation}
\end{theorem}

Hence we can associate every symmetric convex game $(N,v)$ with a function
$s\colon  \left\{0,\dots, n\right\} \to \mathbb{R}$ having the above property.
Similarly, we can apply this to $(N,\K,v)$ with a function
$\sigma\colon \mathcal X \to \mathbb{R}$ where $\mathcal X \subseteq
\{0,\dots,n\}$ is constructed from $\K$. To formalise these constructions, we
define reduced forms of games $(N,v)$ and $(N,\K,v)$.

\begin{definition}\label{def:reduced-forms}
    Let $(N,v)$ be a symmetric game and $(N,\K,v)$ a symmetric
    incomplete game.
    \begin{itemize}
        \item The \emph{reduced form of a game} $(N,v)$ is an ordered pair $(N,s)$, where the function $s\colon \left\{0,\dots,n\right\} \to \mathbb{R}$ is a reduced
        characteristic function such that $s(k) \coloneqq v(S)$ for any $S \subseteq N$ with
        $\lvert S \rvert = k$.
        \item The \emph{reduced form of an incomplete game} $(N,\K,v)$ is a tuple $(N,\mathcal X,\sigma)$ where $\mathcal X = \{i |\, i \in \left\{0,\dots,n\right\}, \exists S \in
        \K: \lvert S \rvert = i\}$ and the function $\sigma\colon \mathcal X \to
        \mathbb{R}$ is defined as $\sigma(k) \coloneqq v(S)$ for any $S \in \K$ such that
        $\lvert S \rvert = k$.
    \end{itemize}
\end{definition}
We also call $(N,s)$ and $(N, \mathcal X, \sigma)$ \emph{the reduced game} and
\emph{the reduced incomplete game}, respectively.

Since $\emptyset$ always belongs to $\K$, for every reduced incomplete game $(N,\mathcal X, \sigma)$, it also holds that $0 \in \mathcal X$ and $\sigma(0) = 0$. When we consider a reduced game $(N,s)$ of a $\Co_\sigma(v)$-extension, we often denote this, for brevity, by $(N,s) \in \Co_\sigma(v)$. By $\overline{X}$, we denote the complement of $\mathcal{X}$ in $\{0,\dots,n\}$, i.e.\ $\overline{\mathcal{X}} \coloneqq \{0,\dots,n\}\setminus \mathcal{X}$.

Notice that a game $(N,v)$ is symmetric convex if and only if the function $s$
of its reduced form $(N,s)$ satisfies property \eqref{eq:convex-sym-char} from
Theorem~\ref{thm:symconv}.

We can visualize the reduced form $(N,s)$ of a symmetric convex game $(N,v)$
by a graph in $\mathbb{R}^2$. On the $x$-axis we put the coalition sizes and
on the $y$-axis the values of~$s$. The point $(0,0)$ is fixed for all reduced games. Now by Theorem~\ref{thm:symconv}, the conditions for $k\in\{1,\dots, n-1\}$ enforce that for $i \in \left\{0,\dots,n\right\}$, points $(i,s(i))$ lie in a \textit{convex} position. More precisely, if we connect the neighbouring pairs $(i,s(i)), (i+1,s(i+1))$ (where $i \in \{0,\ldots,n-1\}$)
by line segments, we obtain a graph of a convex function. The graph is illustrated with an example in Figure~\ref{fig:sym_convex_possibilities}.
Further in this text, we refer to this function as the \emph{line chart} of $(N,s)$. Similarly, for $(N,\mathcal{X},\sigma)$, the line chart is obtained by
connecting consecutive elements from $\mathcal{X}$ by line segments. If $n \in \overline{\mathcal{X}}$, the rightmost line segment is extended to end at $x$-coordinate $n$. The values of $s$ are then set to lie on the union of these line segments.

\begin{figure}
    \centering
    \includegraphics[scale=0.6]{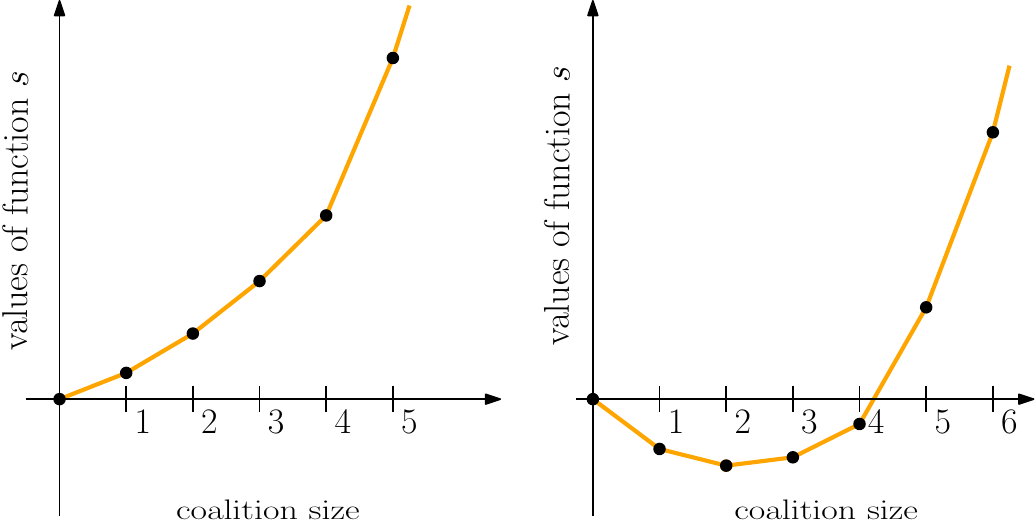}
    \caption{Examples of line charts of symmetric convex games in their reduced forms. The figure on the left depicts a game $(N,s)$ where $s(1) > 0$, the graph on the right a situation where $s(1) < 0$. The slopes of the line segments are bounded by convexity of the function.}
    \label{fig:sym_convex_possibilities}
\end{figure}

\subsubsection{$\Co_\sigma$-extendability}

For an incomplete game in reduced form, i.e.\ $(N,\mathcal X, \sigma)$, the first
question that arises is that of $\Co_\sigma$-extendability. For $ \mathcal X = \left\{0,i\right\}$ with $i \in \left\{1,\dots,n\right\}$, the game is always $\Co_\sigma$-extendable (a possible $\Co_\sigma$-extension is the one where the values of each coalition size lie on the line coming through $(0,\sigma(0))$ and $(i,\sigma(i))$). Therefore, in
the following theorem, we consider $\lvert \mathcal X \rvert > 2$.

\begin{theorem}\label{thm:convex-sym-extendability}
    Let $(N,\mathcal X,\sigma)$ be a reduced form of a symmetric incomplete
    game $(N,\K,v)$ where $\lvert \mathcal X \rvert > 2$. The game is $\Co_\sigma$-extendable if and only if
    \[
    \sigma(k_2) \leq \sigma(k_1) + (k_2 - k_1) \frac{\sigma(k_3) - \sigma(k_1)}{k_3-k_1},
    \]
    for all consecutive elements $k_1 < k_2 < k_3$ from $\mathcal X$.
\end{theorem}

\begin{proof}
    If the game is $\Co_\sigma$-extendable, let $(N,s)$ be the reduced form of any of its $\Co_\sigma$-extension. By Theorem~\ref{thm:symconv}, the line chart of $(N,s)$ is a convex function that coincides with~$\sigma$ on the values of $\mathcal{X}$. Therefore, for any consecutive elements $k_1,k_2,k_3$ from $\mathcal{X}$, the inequality must hold.
    
    For the opposite implication, we construct a $\Co_\sigma(v)$-extension by setting the values of~$s$ to lie on the line chart of $(N,
    \mathcal X, \sigma)$. The construction is illustrated in
    Figure~\ref{fig:sym_conv_segments}.
    
    \begin{figure}
        \centering
        \includegraphics[scale=0.6]{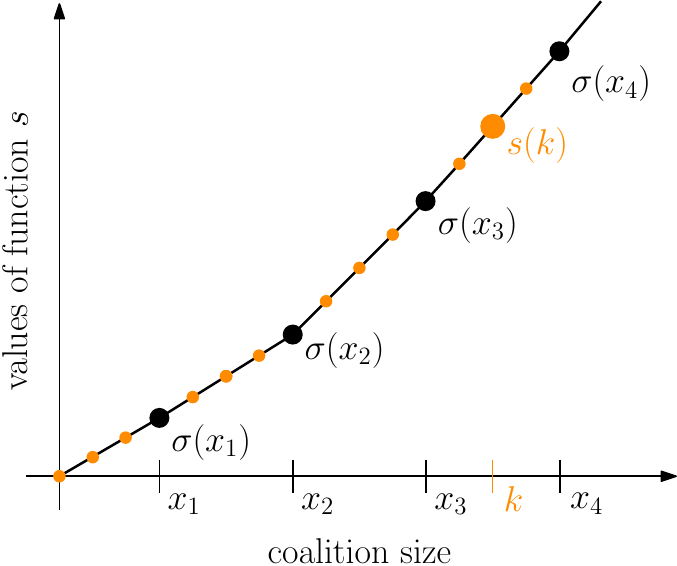}
        \caption{The construction of a $\Co_\sigma$-extension of $(N,\mathcal{X},\sigma)$ where $\mathcal X = \{x_1,x_2,x_3,x_4\}$, using the line chart of $(N,\mathcal X,\sigma)$. The value $s(k)$ lies on the line segment connecting $(x_3,\sigma(x_3))$ and $(x_4,\sigma(x_4))$.}
        \label{fig:sym_conv_segments}
    \end{figure}
    
    Notice that $s(k) = \sigma(k)$ for $k \in \mathcal X$ and also, because the inequalities for consecutive elements $k_1,k_2,k_3$ from $\mathcal X$ hold, the line chart represents a convex function. Thus for all $k \in \{1,\dots,n-1\}$, it holds 
    \[
    s(k)\leq \frac{s(k-1)+s(k+1)}{2}
    \]
    and by Theorem~\ref{thm:symconv}, the game $(N,s)$ is in $\Co_\sigma(v)$.
    \qed \end{proof}

As a direct consequence of the previous theorem, the problem of
$\Co_\sigma$-exten\-da\-bi\-li\-ty of symmetric incomplete games can be decided in linear time with respect to the size of the original
game (i.e.\ the size of the characteristic function).

\subsubsection{The lower game and the upper game}
The following proposition addresses the boundedness of the set of $\Co_\sigma$-ex\-ten\-sions. The restriction to $\lvert N \rvert \geq 3$ is without loss of generality, because for $\lvert N \rvert \leq 2$, when the game $(N,\X,\sigma)$ is not complete and is $\Co_\sigma$-extendable, the set of $\Co_\sigma$-extensions is always unbounded.

\begin{proposition}
    Let $(N,\mathcal{X},\sigma)$ be the reduced form of a $\Co_\sigma$-extendable symmetric incomplete game $(N,\K, v)$ with $\lvert N \rvert \geq 3$. The $\Co_\sigma(v)$ is bounded if and only if $\lvert\mathcal{X}\rvert \geq 3$ and $n \in \mathcal{X}$.
\end{proposition}
\begin{proof}
    Let $(N,\mathcal{X},\sigma)$ be the reduced form of a $\Co_\sigma$-extendable incomplete
    game. If $n \in \overline{\mathcal{X}}$, clearly, from Theorem~\ref{thm:symconv} there
    is no upper bound on the profit of $n$. Let $n \in \mathcal X$ and suppose for a contradiction
    that there is $k \in N$ such that there is no upper bound on its profit.
    Choose a $\Co_\sigma(v)$-extension $(N,s)$ such that $s(k) > k\frac{\sigma(n)}{n}$. The line chart of $(N,s)$ is not a convex function (the property is violated for $(0,s(0)),(k,s(k)),(n,s(n))$), therefore $(N,s)\not\in\Co_\sigma(v)$.
    
    If $\lvert \mathcal{X} \rvert \leq 2$, then $\mathcal{X} = \{0,n\}$ (otherwise
    the set of $\Co_\sigma$-extensions is not bounded from above). Let $\ell$ be a negative
    value smaller than or equal to $\sigma(n)$. Any game $(N,s_{\ell})$ with $s_{\ell}(k)
    = \ell$ for $k \in \{1,\dots,n-1\}$ and $s_{\ell}(0) = \sigma(0), s_{\ell}(n) = \sigma(n)$ is a $\Co_\sigma(v)$-extension of $(N,\X,\sigma)$. Thus, there is no lower bound
    on values of $1,\dots,n-1$.
    
    If $\lvert \mathcal{X} \rvert \geq 3$, then let $i \in
    \mathcal X \setminus \left\{0,n\right\}$. For $k\in \{1,\dots,i-1\}$, the point $(k,s(k))$
    must lie on or above the line coming through points
    $(i,\sigma(i)),(n,\sigma(n))$, otherwise the convexity of line chart of $(N,s)$ is violated, leading to a contradiction. Similarly, for any
    $k\in \{i+1,\dots,n-1\}$ the value $s(k)$ must lie on or above the line coming through points $(0,\sigma(0)),(i,\sigma(i))$, otherwise the convexity is violated, again. The profit of every $k$ is therefore bounded from below.
    \qed \end{proof}

\begin{theorem}\label{thm:convex-sym-lower-upper-game}
    Let $(N,\mathcal X,\sigma)$ be the reduced form of a $\Co_\sigma$-extendable symmetric incomplete game. Suppose that $\Co_\sigma(v)$ is bounded. Furthermore, for every $k \in \overline{\mathcal{X}}$, denote by
    $i_1,i_2,j_1,j_2$ the closest distinct elements from $\mathcal X$ such that it holds\\
    $i_1 < i_2 < k < j_1 < j_2$, if they exist. Then the lower game has the
    following form:
    \[
    \underline{s}(k) \coloneqq \begin{cases}
        \sigma(k), & \text{if } k \in \K,\\
        \sigma(i_1) + (k-i_1)\frac{\sigma(i_2) - \sigma(i_1)}{i_2-i_1}, & \text{if } k\not\in\K \text{ and } j_2 \text{ does not exist,}\\
        \sigma(j_1) + (k-j_1)\frac{\sigma(j_2) - \sigma(j_1)}{j_2-j_1}, & \text{if } k\not\in\K \text{ and } i_1 \text{ does not exist,}\\
        \max\begin{cases*}
            \begin{rcases*}
                \sigma(i_1)+ (k-i_1)\frac{\sigma(i_2) - \sigma(i_1)}{i_2-i_1},\\
                \sigma(j_1) + (k-j_1)\frac{\sigma(j_2) - \sigma(j_1)}{j_2-j_1}\\
            \end{rcases*}
        \end{cases*}, & \text{if } k\not\in\K \text{ and } i_1,i_2,j_1,j_2 \text{ exist.}\\
    \end{cases}
    \]
    The upper game has the following form:
    \[
    \overline{s}(k) \coloneqq \begin{cases}
        \sigma(k), & \text{if } k \in \mathcal{X},\\
        \sigma(i_2) + (k-i_2)\frac{\sigma(j_1) - \sigma(i_2)}{j_1-i_2}, &\text{otherwise.}\\ 
    \end{cases}
    \]
\end{theorem}
\begin{proof}   
    To prove that $(N,\underline{s})$ is the lower game, we start by showing that
    for every $\Co_\sigma$-extension $(N,w)$ and every coalition size $k \in
    N$, it holds that $\underline{s}(k) \leq w(k)$. If $k \in \mathcal{X}$,
    trivially $\underline{s}(k) = \sigma(k) = w(k)$. If $k \notin \mathcal{X}$, then since any $\Co_\sigma$-extension
    must have a convex line chart, the value $w(k)$ must lie on or above the
    lines coming through pairs of points $(i_1,\sigma(i_1)),(i_2,\sigma(i_2))$
    and $(j_1,\sigma(j_1)),(j_2,\sigma(j_2))$. The three cases in the
    definition of the lower game capture this fact by setting the value of
    $\underline{s}(k)$ so that it lies on either one of the lines (if the other
    one does not exist) or on the maximum of both of them.
    
    Now it remains to show that for every $k \in N$, the value $\underline{s}(k)$
    is attained for at least one $\Co_\sigma$-extension. We introduce a
    $\Co_\sigma$-extension $(N,s^{\{a,b\}})$ for consecutive $a,b \in \mathcal{X}$ such that $a < b$, described as
    \[
    s^{\{a,b\}}(\ell) \coloneqq \begin{cases}
        \sigma(\ell), & \text{if } \ell \in \mathcal{X},\\
        \underline{s}(\ell), & \text{if } \ell \notin \mathcal{X} \text{ and } a < \ell < b,\\
        \overline{s}(\ell), & \text{if } \ell \notin \mathcal{X} \text{ and either } \ell < a \text{, or } b < \ell.\\
    \end{cases}
    \] 
    Clearly the game is an extension of $(N,\mathcal{X},\sigma)$. For $i \in \{2,\dots,n-1\}$ such that all three values $s^{\{a,b\}}(i-1),s^{\{a,b\}}(i),s^{\{a,b\}}(i+1)$ coincide with the respective values of the upper game $\overline{s}$, it holds $s^{\{a,b\}}(i) \leq \frac{s^{\{a,b\}}(i-1) +
        s^{\{a,b\}}(i+1)}{2}$, because $(N,\overline{s})$ is a symmetric convex game (as we show further in this proof) so by Theorem~\ref{thm:symconv}, the same inequality holds for values of
    $\overline{s}$. In the rest of the cases, either all the three points
    $(i-1,s^{\{a,b\}}(i-1)),(i,s^{\{a,b\}}(i)),(i+1,s^{\{a,b\}}(i+1))$ lie
    on the same line and the inequality holds with the equal sign, or the
    three points lie on the maximum of two lines coming through pairs of points $(a_2,\sigma(a_2)),(a,\sigma(a))$ and $(b,\sigma(b)),(b_2,\sigma(b_2))$ where $a_2 < a$ and $b < b_2$ are consecutive pairs from $\mathcal{X}$. If $s^{\{a,b\}}(i) >
    \frac{s^{\{a,b\}}(i-1) + s^{\{a,b\}}(i+1)}{2}$, then either $\sigma(a) >
    \sigma(a_2) + (a - a_2) \frac{\sigma(b) - \sigma(a_2)}{b-a_2}$ or
    $\sigma(b) > \sigma(a) + (b - a) \frac{\sigma(b_2) - \sigma(a)}{b_2- a}$, both resulting, by Theorem~\ref{thm:convex-sym-extendability}, in a contradiction with the $\Co_\sigma$-ex\-ten\-da\-bi\-li\-ty of
    $(N,\mathcal{X},\sigma)$. Now for $k \in \mathcal{X}$, we choose $(N,s^{\{a,b\}})$ such that $a = k$ and for $k \notin \mathcal X$, we choose $(N,s^{\{a,b\}})$ such that $ a < k < b$ are the closest coalition sizes with defined value.
    
    For the upper game $(N,\overline{s})$, suppose for a contradiction that there
    is the reduced form $(N,s)$ of a $\Co_\sigma(v)$-extension such that for $k \in N$, $\overline{s}(k) < s(k)$. As for $k \in \mathcal{X}$, $\overline{s}(k) = \sigma(k) = s(k)$, it must be that $k \notin \mathcal{X}$. But if $k \notin \mathcal{X}$ and $\overline{s}(k) = \sigma(i_2) + (k-i_2)\frac{\sigma(j_1) - \sigma(i_2)}{j_1-i_2} < s(k)$, the convexity of the line chart is violated, because $(k,s(k))$ lies above the line segment between points $(i_2,\sigma(i_2))$,$(j_1,\sigma(j_1))$. This is a contradiction.
    
    Now we prove that $(N,\overline{s})$ is a $\Co_\sigma$-extension of
    $(N,\mathcal X,\sigma)$. First, it is clearly an extension. Furthermore, notice that the values of $(N,\overline{s})$ lie on the line chart of $(N,\mathcal{X},\sigma)$. Since the game is $\Co_\sigma$-extendable, the line chart is a convex function, therefore inequalities \eqref{eq:convex-sym-char} from Theorem~\ref{thm:symconv} hold, meaning $(N,\overline{s})\in \Co_\sigma(v)$.
    \qed \end{proof}

The game $(N,\overline{s})$ is always a $\Co_\sigma$-extension, however, this is not true for $(N,\underline{s})$ in general, as can be seen in the example in Figure~\ref{fig:sym_conv_lower}.
\begin{figure}
    \centering
    \includegraphics[scale=0.6]{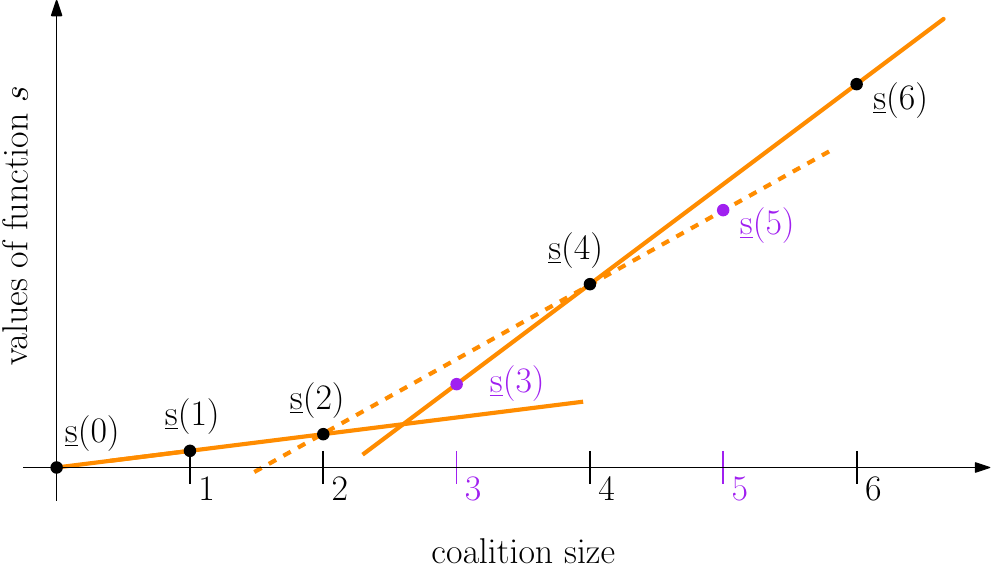}
    \caption{An example of a reduced game $(N,\mathcal X,\sigma)$ with $\mathcal X = \left\{0,1,2,4,6\right\}$ where the condition $\frac{\underline{s}(3) + \underline{s}(5)}{2} \ngeq \underline{s}(4)$ from Theorem~\ref{thm:symconv} is not satisfied. This implies that $(N,\underline{s})$ is not a $\Co_\sigma$-extension of $(N,\X,\sigma)$.}
    \label{fig:sym_conv_lower}
\end{figure}
\subsubsection{Extreme games}

Games $(N,s^{\{a,b\}})$ are actually even more important because they are extreme games of $\Co_\sigma(v)$. 

\begin{proposition}
    Let $(N,\mathcal{X},\sigma)$ be the reduced form of a $\Co_\sigma$-extendable symmetric incomplete game $(N,\K,v)$. Games $(N,s^{\left\{a,b\right\}})$ for consecutive $a,b \in \mathcal{X}$, where $a < b$, and $(N,\overline{s})$, are extreme games of $\Co_\sigma(v)$.
\end{proposition}

\begin{proof}
    For a contradiction, suppose that for some $a,b$, $(N,s^{\{a,b\}})$ is not an extreme game of $\Co_\sigma(v)$. By Definition~\ref{def:extreme-points}, there are two $\Co_\sigma(v)$-extensions $(N,s_1)$ and $(N,s_2)$ such that $(N,s^{\{a,b\}})$ is their nontrivial convex combination and without loss of
    generality, there is $i \in \{0,\dots,n\}$ such that $s_1(i) < s^{\{a,b\}}(i)
    < s_2(i)$. For $i \in \mathcal{X}$, this is not possible as $s_1(i) =
    s_2(i) = s^{\{a,b\}}(i)$. Furthermore, for $i \notin \mathcal{X}$ and $a < i < b$,
    this is a contradiction with $s_1(i) < s^{\{a,b\}}(i) = \underline{s}(i)$ and
    finally for $i \notin \mathcal{X}$ and either $i < a$ or $b < i$, we get again a
    contradiction because $\overline{s}(i) = s^{\{a,b\}}(i) < s_2(i)$. Following
    a similar argument, we conclude that the upper game $(N,\overline{s})$ is also an extreme
    game.
    \qed \end{proof}

In general, $(N,\overline{s})$ and $(N,s^{\{a,b\}})$ are not the only extreme games. In the following theorem, we describe all the extreme games of $\Co_\sigma(v)$.

\begin{theorem}
    Let $(N,\mathcal{X},\sigma)$ be the reduced form of a $\Co_\sigma$-extendable symmetric incomplete game such that $\Co_\sigma(v)$ is bounded. For $k \in \left\{0,\dots,n\right\} \setminus \mathcal{X}$ and $i,j \in \mathcal{X}$ closest to $k$ such that $i < k < j$, games $(N,s^k)$ defined as
    \[
    s^k(m) \coloneqq \begin{cases}
        \sigma(m), & \text{if } m \in \mathcal{X},\\
        \overline{s}(m), & \text{if } m \notin \mathcal{X} \text{ and either } m < i \text{ or } j < m,\\
        \underline{s}(m), & \text{if } m = k,\\
        \sigma(j) + (m - j)\frac{\sigma(j)-\underline{s}(k)}{j - k}, & \text{if } m \notin \mathcal{X} \text{ and } k < m < j,\\
        \sigma(i) + (m - i)\frac{\underline{s}(k) - \sigma(i)}{k - i}, & \text{if } m \notin \mathcal{X} \text{ and } i < m < k\\
    \end{cases}
    \]
    together with $(N,\overline{s})$ form all the extreme games of $\Co_\sigma(v)$.
\end{theorem}
\begin{proof}
    We divide the proof into two parts. In the first part, we show that any $\Co_\sigma(v)$-extension $(N,s)$ is a convex combination of games $(N,\overline{s})$ and $(N,s^k)$ for $k \in \overline{\mathcal{X}}$. In the second part, we show that every game $(N,s^k)$ is an extreme game, thus (together with the upper game $(N,\overline{s})$) they form all the extreme games.
    
    Before we begin, let us define a \emph{gap} as an inclusion-wise maximal nonempty sequence of consecutive coalition sizes with undefined profit. In other words, we can say that there is a gap between $i$ and $j$ if $i,j \in \mathcal X$, $i < j$, $j-i > 1$, and for every $i'$ such that $i < i' < j$, it holds $i' \in \overline{\mathcal{X}}$. The \emph{size} of the gap between $i$ and $j$ is defined as $j-i-1$, that is the number of coalition sizes with unknown values in the given gap. It is immediate that the size of every gap is at least one.
    
    We now prove the first part of the theorem. First, let us suppose that there is only one gap in $(N,\X,\sigma)$. We prove this case by induction on the size of the gap.
    
    If the size of the gap is 1, there is only one game $(N,s^k)$ that is equal to $(N,s^{\{k-1,k+1\}})$. Any $\Co_\sigma$-extension $(N,s)$ can be expressed as a convex combination of this game and the upper game $(N,\overline{s})$ as $s = \alpha s^k +
    (1-\alpha)\overline{s}$ with 
    \[\alpha = \frac{s(k)-\overline{s}(k)}{s^k(k) - \overline{s}(k)} \in[0,1].\]
    
    For the induction step, suppose that the size of the gap between $i$ and $j$ is $\ell$, $\ell > 1$. Hence there are $\ell$ games 
    \[
    (N,s^{i+1}),(N,s^{i+2}),\dots,(N,s^{j-1}) \text{ together with } (N,\overline{s}).
    \]
    We construct a new system of $\ell-1$ games
    \[
    (N,(s^{i+2})'),(N,(s^{i+3})'),\dots,(N,(s^{j-1})')\text{ together with } (N,(s^{i+1})'),
    \]
    \[
    \text{ where } (s^{m})' \coloneqq \alpha s^{m} + (1-\alpha)\overline{s}\text{ and }\alpha = \frac{s(i+1)-\overline{s}(i+1)}{s^{i+1}(i+1) - \overline{s}(i+1)}.
    \]
    These games correspond to the extreme games of an incomplete game $(N,\mathcal X',\sigma')$ where $\mathcal X' \coloneqq \mathcal{X} \cup \{i+1\}$, and the function $\sigma'$ is defined as $\sigma'(m) \coloneqq \sigma(m)$ for $m \in \mathcal{X}$ and $\sigma'(i+1) \coloneqq s(i+1)$. The game $(N,(s^{i+1})')$ represents the upper game of $\Co_\sigma(v)$.
    Since the new system of $\ell$ games forms the extreme games of $\Co_{\sigma'}$-extensions of $(N,\mathcal{X}',\sigma')$, the game $(N,s)$ (which is also a $\Co_{\sigma'}$-extension of $(N,\mathcal{X}',\sigma')$) is, by induction hypothesis, their convex combination. And as each game $(N,(s^{m})')$ is a convex combination of $(N,\overline{s})$ and $(N,s^{m})$, the game $(N,s)$ is also a convex combination of the former system \[
    (N,s^{i+1}),(N,s^{i+2}),\dots,(N,s^{j-1}) \text{ together with } (N,\overline{s}).
    \]
    
    Notice that if there is more than one gap between the coalition sizes in $\mathcal{X}$, then we can follow a similar construction as in the situation with precisely one gap. This is because any two extreme games parametrised by two coalition sizes from one gap assign the same profit to any coalition size from a different gap. Thus, we can start our construction by filling in the first gap, after that, taking the extreme games of the extended incomplete game and so on, until there is no gap left.
    
    As for the second part of the proof, suppose for a contradiction that $(N,s^k)$ for $k \in \overline{\mathcal{X}}$ is not an extreme game of $\Co_\sigma(v)$. By Definition~\ref{def:extreme-points}, there are $\Co_\sigma$-extensions $(N,s_1),(N,s_2)$ and $m \in N$ such that $s_1(m) < s^k(m) < s_2(m)$. Clearly, $m \notin \mathcal{X}$ (since $s_1(m) = s^k(m) = s_2(m) = \sigma(m)$) and if $m$ is such that $s^k(m) = \overline{s}(k)$ or $s^k(m) = \underline{s}(m)$, we arrive at a contradiction. Therefore, the only case that remains is $m \notin \mathcal{X}$ together with $i < m < j$ and $m \neq k$. For any such $m$, the convexity of the line chart is violated either for $(i,s_1(i)),(k,s_1(k)),(m,s_1(m))$ (if $k < m$), or for $(i,s_2(i)),(m,s_2(m)),(k,s_2(k))$ (if $m < k$). Both cases are depicted in Figure~\ref{fig:sym_conv_extreme_games}.
\qed \end{proof}

\begin{figure}
    \centering
    \includegraphics[width=\textwidth]{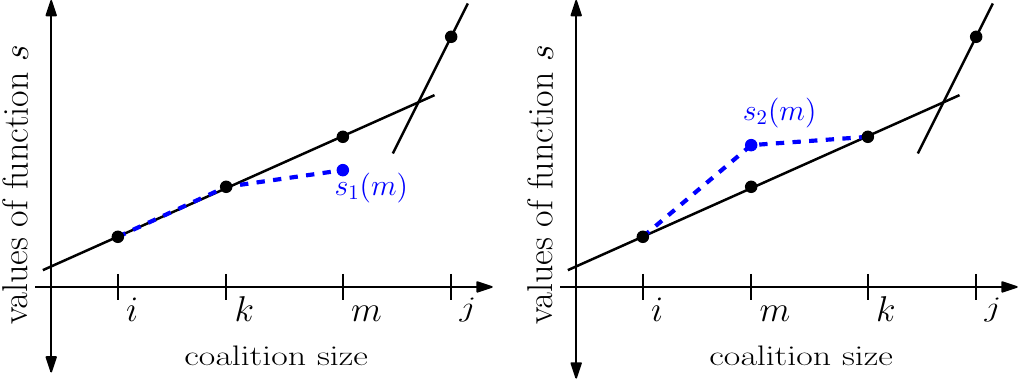}
    \caption{Examples of a violation of convexity of the line chart of both $(N,s_1)$ and $(N,s_2)$. The full lines depict the line chart of $(N,\underline{s})$ and the dotted lines depict the line charts of $(N,s_1)$ and $(N,s_2)$. On the left, the situation where $k < m$ is shown. We have values $s^k(i) = s_1(k)$ and $s^k(k)=s_1(k)$, yet $s_1(m)$ is too small. Similarly, on the right, the situation where $m < k$ is shown, with $s^k(i) = s_2(k)$, $s^k(k) = s_2(k)$. However, in this case, the value $s_2(m)$ is too big.} 
    \label{fig:sym_conv_extreme_games}
\end{figure}

For a $\Co_\sigma$-extendable symmetric incomplete game in a reduced form $(N,\mathcal
X,\sigma)$ with $\Co_\sigma(v)$ bounded and $\lvert \Co_\sigma(v) \rvert > 1$, the number of extreme games is always $\lvert \overline{\mathcal{X}} \rvert
+ 1 = n - \lvert \mathcal{X} \rvert + 2$, no matter what the values of $\sigma$ are.

Algebraically, we can describe the set $\Co_\sigma(v)$ as
\begin{equation}\label{eq:convex-sym-description}
    \Co_\sigma(v) = \Bigg\{\Big(N,\overline{\alpha} \text{ }\overline{s} + \sum_{k \in \overline{\mathcal{X}}} \alpha_k s^k\Big) \bigg\vert\,\, \overline{\alpha} + \sum_{k \in \overline{\mathcal{X}}} \alpha_k = 1, \overline{\alpha}, \alpha_k \geq 0, k \in \overline{\mathcal{X}}\Bigg\},
\end{equation}
namely as the set of convex combinations of extreme games $\overline{s}$ and $s^k$ for $k \in \overline{\mathcal{X}}$.

Geometrically, we can describe the set $\Co_\sigma(v)$ when we restrict the game $(N,\X,\sigma)$ a little. First, suppose $\mathcal{X} = \{0,n\}$ and $\sigma(0) = \sigma(n) = 0$. According to Theorem~\ref{thm:symconv}, we can describe $\Co_\sigma(v)$ by a system of $n-1$ inequalities with $n-1$ unknowns, $Ay \leq 0$, where 
\[
\addtolength{\arraycolsep}{0.5ex}
A=\begin{pmatrix}
    2 & -1 & & 0\\ -1 & \ddots & \ddots &  \\& \ddots  & \ddots & -1\\0 & & -1 & 2
\end{pmatrix}.
\]
The matrix $A$ is an \textit{M-matrix}~\cite{Horn1991}, therefore it is
nonsingular and $A^{-1} \geq 0$. Nonsingularity of $A$ implies that
$\Co_\sigma(v)$ is a pointed polyhedral cone, which is translated such that its
vertex is not necessarily in the origin of the coordinate system. Furthermore,
because $A^{-1} \geq 0$, the \textit{normal cone} $\Co_\sigma(v)^{*}$ of
$\Co_\sigma(v)$ (see~\cite{Boyd2004}) contains the whole nonnegative orthant.
Thus, the vertex of polyhedral cone $\Co_\sigma(v)$ is the biggest element of
$\Co_\sigma(v)$ when restricted to each coordinate (this corresponds with the
statement that the upper game is a $\Co_\sigma$-extension). Therefore,
geometrically, the set $\Co_\sigma(v)$ looks like \textit{squeezed} negative
orthant. For an incomplete game $(N,\mathcal{X'},\sigma')$ where $\{0,n\}
\subseteq \mathcal{X'}$ and $\sigma'(0)=\sigma(n)=0$, the set of $\Co_\sigma$-extensions is $\Co_\sigma(v)$ with some of the coordinates fixed, i.e.
\[
\Co_\sigma(v) \cap_{k\in \mathcal{X'}}\{s(k)=\sigma(k)\}.
\]

\section{Conclusion} \label{sec:conclusion}

We would like to conclude our paper with observations on a connection between
incomplete cooperative games and cooperative interval games. This
connection has not been mentioned in literature so far and we think it
provides a nice bridge between the two approaches to uncertainty in
cooperative game theory.

We remind the reader that a \emph{cooperative interval game} is a pair
$(N,w)$, where $N$ is a finite set of players and $w\colon 2^N \to
\mathbb{IR}$ is the characteristic function of this game with $\mathbb{IR}$
being the set of all real closed intervals. We further set $v(\emptyset) \coloneqq
[0,0]$.

\begin{proposition}
    For a given incomplete game $(N,\mathcal K,v)$ and a set of its extensions
    $\mathcal E$, the associated lower game and upper game induce a cooperative
    interval game $(N,w)$ containing the set $\mathcal E$ of extensions of
    $(N,\mathcal K,v)$. Furthermore, this game is inclusion-wise minimal, i.e.\ 
    for every $S \subseteq N$, there is an extension from $\mathcal E$ attaining
    the lower bound of $w(S)$ and an extension from $\mathcal E$ attaining the
    upper bound of $w(S)$.
\end{proposition}

We can also take another view-point. We can generalise the definition of
incomplete game to the interval setting by allowing the partial game to be an
interval game. We can then ask what are the extensions having some desired
property, for example being selection superadditive interval games. Indeed,
this aligns with the main motivation behind some of the results in~\cite{Bok2015}
and~\cite{1811.04063}. 

We think that this is just the first step towards unifying both theories
together. The fact that so far, no connection has been made between these two
areas in literature seems surprising to us. Some of the issues raised here are
work in progress.

\addcontentsline{toc}{section}{References}
\bibliographystyle{elsarticle-num}
\bibliography{bibliography}

\begin{thebibliography}{10}

\bibitem{Gok2009a}
S.~Z. Alparslan~G{\"o}k.
\newblock {\em Cooperative interval games}.
\newblock PhD thesis, Middle East Technical University, 2009.

\bibitem{Gok2011}
S.~Z. Alparslan~G{\"o}k, O.~Branzei, R.~Branzei, and S.~Tijs.
\newblock Set-valued solution concepts using interval-type payoffs for interval
  games.
\newblock {\em Journal of Mathematical Economics}, 47(4):621--626, 2011.

\bibitem{Gok2009b}
S.~Z. Alparslan~G{\"o}k, S.~Miquel, and S.~H. Tijs.
\newblock Cooperation under interval uncertainty.
\newblock {\em Mathematical Methods of Operations Research}, 69(1):99--109,
  2009.

\bibitem{Bhaskar2019}
U.~Bhaskar and G.~Kumar.
\newblock {The Complexity of Partial Function Extension for Coverage
  Functions}.
\newblock In Dimitris Achlioptas and L{\'a}szl{\'o}~A. V{\'e}gh, editors, {\em
  Approximation, Randomization, and Combinatorial Optimization. Algorithms and
  Techniques (APPROX/RANDOM 2019)}, volume 145 of {\em Leibniz International
  Proceedings in Informatics (LIPIcs)}, pages 30:1--30:21, Dagstuhl, Germany,
  2019. Schloss Dagstuhl--Leibniz-Zentrum fuer Informatik.
\newblock URL: \url{http://drops.dagstuhl.de/opus/volltexte/2019/11245}, \href
  {https://doi.org/10.4230/LIPIcs.APPROX-RANDOM.2019.30}
  {\path{doi:10.4230/LIPIcs.APPROX-RANDOM.2019.30}}.

\bibitem{Bhaskar2018}
U.~Bhaskar and G.~Kumar.
\newblock {Partial Function Extension with Applications to Learning and
  Property Testing}.
\newblock In Yixin Cao, Siu-Wing Cheng, and Minming Li, editors, {\em 31st
  International Symposium on Algorithms and Computation (ISAAC 2020)}, volume
  181 of {\em Leibniz International Proceedings in Informatics (LIPIcs)}, pages
  46:1--46:16, Dagstuhl, Germany, 2020. Schloss Dagstuhl--Leibniz-Zentrum
  f{\"u}r Informatik.
\newblock URL: \url{https://drops.dagstuhl.de/opus/volltexte/2020/13390}, \href
  {https://doi.org/10.4230/LIPIcs.ISAAC.2020.46}
  {\path{doi:10.4230/LIPIcs.ISAAC.2020.46}}.

\bibitem{Bilbao2012}
J.~M. Bilbao.
\newblock {\em Cooperative Games on Combinatorial Structures}, volume~26 of
  {\em Theory and Decision Library}.
\newblock Springer Science \& Business Media, 2012.

\bibitem{1811.04063}
J.~Bok.
\newblock On convexity and solution concepts in cooperative interval games.
\newblock {\em arXiv preprint arXiv:1811.04063}, 2018.

\bibitem{Bok2015}
J.~Bok and M.~Hlad\'{i}k.
\newblock Selection-based approach to cooperative interval games.
\newblock In {\em Communications in Computer and Information Science, ICORES
  2015 - International Conference on Operations Research and Enterprise
  Systems, Lisbon, Portugal, 10-12 January, 2015}, volume 577, pages 40--53,
  2015.

\bibitem{BC1convex}
J.~Bok and M.~\v{C}ern\'{y}.
\newblock 1-convex extensions of incomplete cooperative games and the average
  value.
\newblock {\em arXiv preprint arXiv:2107.04679}, 2022.

\bibitem{Bondareva1963}
O.~N. Bondareva.
\newblock Some applications of linear programming methods to the theory of
  cooperative games.
\newblock {\em Problemy kibernetiki}, 10:119--139, 1963.

\bibitem{Boyd2004}
S.~Boyd and L.~Vandenberghe.
\newblock {\em Convex Optimization}.
\newblock Cambridge University Press, 2004.

\bibitem{Branzei2008}
R.~Branzei, D.~Dimitrov, and S.~Tijs.
\newblock {\em Models in Cooperative Game Theory}, volume 556 of {\em Lecture
  Notes in Economics and Mathematical Systems}.
\newblock Springer, 2008.

\bibitem{Chateauneuf1989}
A.~Chateauneuf and J.-Y. Jaffray.
\newblock Some characterizations of lower probabilities and other monotone
  capacities through the use of {M}{\"o}bius inversion.
\newblock {\em Mathematical Social Sciences}, 17:263--283, 1989.

\bibitem{Choquet1954}
G.~Choquet.
\newblock Theory of capacities.
\newblock {\em Annales de l'Institut Fourier}, 5:131--295, 1954.
\newblock \href {https://doi.org/10.5802/aif.53} {\path{doi:10.5802/aif.53}}.

\bibitem{Curiel2013}
I.~Curiel.
\newblock {\em Cooperative Game Theory and Applications: {C}ooperative Games
  Arising from Combinatorial Optimization Problems}.
\newblock Springer, Dordrecht, 2013.

\bibitem{Driessen1988}
T.~Driessen.
\newblock {\em Cooperative Games, Solutions and Applications}, volume~3 of {\em
  Theory and Decision Library C}.
\newblock Kluwer, Dordrecht, 1988.

\bibitem{Farkas1902}
J.~Farkas.
\newblock Theorie der einfachen ungleichungen.
\newblock {\em Journal f{\"u}r die reine und angewandte Mathematik},
  1902(124):1--27, 1902.

\bibitem{Fujimoto2005}
K.~Fujimoto and T.~Murofushi.
\newblock Some characterizations of $k$-monotonicity through the bipolar
  {M}{\"o}bius transform in bi-capacities.
\newblock {\em Journal of Advanced Computational Intelligence and Intelligent
  informatics}, 9(5):484--495, 205.

\bibitem{Gilboa1994}
I.~Gilboa and D.~Schmeidler.
\newblock Additive representations of non-additive measures and the {Choquet}
  integral.
\newblock {\em Annals of Operations Research}, 52(1):43--65, 1994.

\bibitem{Gilles2010}
R.~P. Gilles.
\newblock {\em The Cooperative Game Theory of Networks and Hierarchies},
  volume~44 of {\em Theory and Decision Library C}.
\newblock Springer, 2010.

\bibitem{Grabisch2016}
M.~Grabisch.
\newblock {\em Set Functions, Games and Capacities in Decision Making}.
\newblock Springer, 2016.

\bibitem{Grabisch2000}
M.~Grabisch, J.~L. Marichal, and M.~Roubens.
\newblock Equivalent representations of set functions.
\newblock {\em Mathematics of Operations Research}, 25(2):157--178, 2000.

\bibitem{Horn1991}
R.~A. Horn and Ch.~R. Johnson.
\newblock {\em Topics in Matrix Analysis}.
\newblock Cambridge University Press, 1991.

\bibitem{Lemaire1991}
J.~Lemaire.
\newblock {\em Cooperative Game Theory and its Insurance Applications}.
\newblock Center for Research on Risk and Insurance, Wharton School of the
  University of Pennsylvania, 1991.

\bibitem{Mallozzi2011}
L.~Mallozzi, V.~Scalzo, and S.~Tijs.
\newblock Fuzzy interval cooperative games.
\newblock {\em Fuzzy Sets and Systems}, 165(1):98--105, 2011.

\bibitem{Mares2001}
M.~Mare{\v{s}}.
\newblock {\em Fuzzy Cooperative Games: Cooperation with Vague Expectations},
  volume~72.
\newblock Physica-Verlag Heidelberg, 2001.

\bibitem{Mares2004}
M.~Mare{\v{s}} and M.~Vlach.
\newblock Fuzzy classes of cooperative games with transferable utility.
\newblock {\em Scientiae Mathematicae Japonica}, 2:269--278, 2004.

\bibitem{masuya2021approximated}
S.~Masuya.
\newblock An approximated {S}hapley value for partially defined cooperative
  games.
\newblock {\em Procedia Computer Science}, 192:100--108, 2021.

\bibitem{Masuya2016a}
S.~Masuya and M.~Inuiguchi.
\newblock A fundamental study for partially defined cooperative games.
\newblock {\em Fuzzy Optimization Decision Making}, 15(1):281--306, 2016.

\bibitem{Palanci2015}
O.~Palanc{\i}, S.~Z. Alparslan~G{\"o}k, S~Erg{\"u}n, and G.W. Weber.
\newblock Cooperative grey games and the grey {S}hapley value.
\newblock {\em Optimization}, 64(8):1657--1668, 2015.

\bibitem{Palanci2014}
O.~Palanc{\i}, S.~Z. Alparslan~G{\"o}k, and G.~W. Weber.
\newblock Cooperative games under bubbly uncertainty.
\newblock {\em Mathematical Methods of Operations Research}, 80(2):129--137,
  2014.

\bibitem{Peleg2007}
B.~Peleg and P.~Sudh{\"o}lter.
\newblock {\em Introduction to the Theory of Cooperative Games}, volume~34 of
  {\em Theory and Decision Library}.
\newblock Springer Science \& Business Media, 2nd edition, 2007.

\bibitem{Rota1964}
G.~C. Rota.
\newblock On the foundations of combinatorial theory {I}. {T}heory of
  {M}{\"o}bius functions.
\newblock {\em Zeitschrift f{\"u}r Wahrscheinlichkeitstheorie und verwandte
  Gebiete}, 2(4):340--368, 1964.

\bibitem{Seshadhri2014}
C.~Seshadhri and J.~Vondr{\'a}k.
\newblock Is submodularity testable?
\newblock {\em Algorithmica}, 69(1):1--25, 2014.

\bibitem{Shapley1953}
L.~S. Shapley.
\newblock A value for $n$-person game.
\newblock {\em Annals of Mathematical Studies}, 28:307--317, 1953.

\bibitem{Shapley1967}
L.~S. Shapley.
\newblock On balanced sets and cores.
\newblock {\em Naval Research Logistics Quarterly}, 14(4):453--460, 1967.

\bibitem{Shapley1971}
L.~S. Shapley.
\newblock Cores of convex games.
\newblock {\em International Journal of Game Theory}, 1(1):11--26, 1971.

\bibitem{Weber2010}
G.~W. Weber, R.~Branzei, and S.~Z. Alparslan~G{\"o}k.
\newblock On cooperative ellipsoidal games.
\newblock In {\em 24th Mini EURO Conference-On Continuous Optimization and
  Information-Based Technologies in the Financial Sector, MEC EurOPT}, pages
  369--372, 2010.

\bibitem{Willson1993}
S.~J. Willson.
\newblock A value for partially defined cooperative games.
\newblock {\em International Journal of Game Theory}, 21(4):371--384, 1993.

\bibitem{xiaohui2021extension}
Y.~Xiaohui.
\newblock Extension of owen value for the game with a coalition structure under
  the limited feasible coalition.
\newblock {\em Soft Computing}, 25(8):6139--6156, 2021.

\end{thebibliography}

\appendix
\section{Omitted proofs}

\begin{proof}{\textit{(of Proposition \ref{prop:symconvchar})}}
    If the game is symmetric convex, we consider the characterisation from Theorem~\ref{thm:convchar} for coalitions $S,S\cup j$ and $i \in S$, obtaining
    \begin{equation}\label{eq:convex-symmetric2}
        v(S) - v(S\setminus i) \leq v(S\cup j) - v(S\cup j\setminus i).
    \end{equation}
    Because $\lvert (S\cup j)\setminus i \vert = \lvert S \rvert$, we have $v((S\cup j)\setminus i) = v(S)$ by symmetry. By adding $v(S)$ to \eqref{eq:convex-symmetric2} and rearranging the inequality, we get \eqref{eq:convex-symmetric}
    \[v(S) \leq \frac{v(S\setminus i) + v(S\cup j)}{2}.\]
    
    For the opposite implication, suppose that conditions \eqref{eq:convex-symmetric} hold and $(N,v)$ is not convex. Then there is a player $k \in N$ and coalitions $T_1 \subsetneq T_2 \subseteq N \setminus k$ for which the condition from Theorem~\ref{thm:convchar} is violated, i.e. 
    \begin{equation}\label{eq:convchar}
        v(T_1 \cup k) - v(T_1) > v(T_2 \cup k) - v(T_2).    
    \end{equation}
    We choose player $k$ and coalitions $T_1,T_2$ such that the difference $\lvert T_2\rvert - \lvert T_1\rvert$ is minimal. We distinguish two possible cases.
    \begin{enumerate}
        \item If $\lvert T_2\rvert - \lvert T_1\rvert = 1$, then by symmetry of $v$, we have that $v(T_2) = v(T_1 \cup k)$. In that case, we get
        $$v(T_2) > \frac{v(T_1) + v(T_2 \cup  k)}{2}.$$
        Furthermore, there exists a unique $\ell \in T_2 \setminus T_1$ such that $T_1 \cup  \ell = T_2$. Thus we can write
        $$v(T_2) > \frac{v(T_2 \setminus  \ell) + v(T_2 \cup k)}{2},$$
        which leads to a contradiction with \eqref{eq:convex-symmetric}.
        \item If $\lvert T_2\rvert - \lvert T_1\rvert > 1$, then there is a coalition $T_3$ such that $T_1 \subsetneq T_3 \subsetneq T_2 \subseteq N \setminus k$. By minimality of $\lvert T_2 \rvert - \lvert T_1 \rvert$, we know that
        \begin{equation}\label{eq:convex-symmetric3}
        v(T_1 \cup k) - v(T_1) \leq v(T_3 \cup k) - v(T_3)
        \end{equation}
        and
        \begin{equation}\label{eq:convex-symmetric4}
            v(T_3 \cup  k) - v(T_3) \leq v(T_2 \cup  k) - v(T_2).
        \end{equation}
        By adding \eqref{eq:convex-symmetric3} and \eqref{eq:convex-symmetric4} together, we get
        $$v(T_1 \cup  k) - v(T_1) \leq v(T_2 \cup  k) - v(T_2),$$
        which is a contradiction with \eqref{eq:convchar}. \qed 
    \end{enumerate}
\end{proof}

\end{document}